\documentclass[11pt]{article}

\usepackage[lined,algonl,boxed]{algorithm2e}
\usepackage{amsmath,amsfonts,amssymb,amsthm,dsfont}
\usepackage{booktabs}
\usepackage{epsfig}
\usepackage{latexsym}
\usepackage{graphicx}

\usepackage[default]{jasa_harvard}    % 	for formatting citations in text
\usepackage{JASA_manu}

\renewcommand{\harvardand}{and}

\evensidemargin \oddsidemargin \marginparwidth 0.5in
\newtheorem{theorem}{Theorem}[section]
\newtheorem{cor}[theorem]{Corollary}
\newtheorem{lem}[theorem]{Lemma}
\newtheorem{prop}[theorem]{Proposition}
\theoremstyle{definition}
\newtheorem{defn}[theorem]{Definition}

\numberwithin{equation}{section}

\begin{document}
\title{Fast Computing for Distance Covariance}
\author{Xiaoming Huo \\
School of Industrial and Systems Engineering, \\
Georgia Institute of Technology,\\
Atlanta, GA 30332 \\
and National Science Foundation, \\
Arlington, VA 22203 \\
email: \texttt{xiaoming@isye.gatech.edu} \\
\& G\'abor J. Sz\'ekely \\
National Science Foundation \\
Arlington, VA 22203 \\
and Alfr\'ed R\'enyi Institute of Mathematics \\
email: \texttt{gszekely@nsf.gov}
}

%\affiliation{National Science Foundation and Alfr\'ed R\'enyi Institute of Mathematics}

\date{\today}
\maketitle

\newpage

\mbox{}
%\vspace*{2in}
\begin{center}
\textbf{Author's Footnote:}
\end{center}
Dr. Xiaoming Huo is a professor in the School of Industrial and Systems Engineering at
the Georgia Institute of Technology.
He has been serving as a rotator at the National Science Foundation since August 2013.
Mailing address: 765 Ferst Dr, Atlanta, GA 30332
(email: xiaoming@isye.gatech.edu).

Dr. G\'abor J. Sz\'ekely is a program officer at the National Science Foundation.
Mailing address: 4201 Wilson Blvd, Arlington, VA 22203
(email: gszekely@nsf.gov).

\begin{abstract}
Distance covariance and distance correlation have been widely adopted in measuring dependence of a pair of random variables or random vectors.
If the computation of distance covariance and distance correlation is implemented directly accordingly to its definition then its computational complexity
is O($n^2$) which is a disadvantage compared to other faster methods.
In this paper we show that the computation of distance covariance and distance correlation
of real valued random variables can be implemented by an O($n \log n$) algorithm and this is comparable to other computationally efficient algorithms. The new
formula we derive for an unbiased estimator for squared distance covariance turns out to be a U-statistic. This fact implies some nice asymptotic properties
that were derived before via more complex methods.
We apply the fast computing algorithm to some synthetic data.
Our work will make distance correlation
applicable to a much wider class of applications.
\end{abstract}

\vspace*{.3in}


\newpage

\section{Introduction}

Since its induction \cite{Szekely2007},
distance correlation has had many applications in, e.g.,
life science \cite{Kong2012} and
variable selection \cite{Li2012}, and has
been analyzed \cite{Szekely2012,Lyons2013},
extended \cite{Szekely2009,pdc2014} in various aspects.
If distance correlation were implemented straightforwardly from its definition, its computational
complexity can be as high as a constant times $n^2$ for a sample size $n$.
This fact has been cited for numerous times in the literature
as a disadvantage of adopting the distance correlation.
In this paper, we demonstrate that an O($n \log(n)$) algorithm for a version of the distance correlation exits.

The main idea behind the proposed algorithm is to use an idea rooted in the
the AVL tree structure \cite{AVL1962}.
The same idea has been utilized to develop fast algorithm for computing the Kendall's $\tau$ rank correlation coefficient \cite{Knight1966} \cite{Christensen2005}.
We extend it to make it suitable for our purpose.
The derivation of the fast algorithm also involves significant reformulation from the original
version of the distance correlation.
Details are presented in this paper.

In simulations, not only we demonstrate the effectiveness of the fast algorithm, but also we testify
that the advantage of using distance correlation (in comparison with other existing methods) become more
evident when the sample sizes increase.
These experiments become feasible due to the availability of the proposed fast algorithm.
In one experiment (See details in Section \ref{sec:screening}), we increased the sample
size by $100$ fold from a previously published simulation study.

The rest of this paper is organized as follows.
Section \ref{sec:reviewDistCorr} reviews the distance covariance/correlation and
its relevant properties.
In Section \ref{sec:unbiased}, we consider a reformulation of the distance
covariance, such that the new estimator is both unbiased and a U-statistic.
In Section \ref{sec:algorithm}, an algorithm with the average
complexity of $O(n \log n)$ was presented.
Extensive simulations are presented in Section \ref{sec:simulations} to
demonstrate the additional capability we obtained due to the proposed
fast algorithm.
Finally, some concluding remarks are made in Section \ref{sec:conclude}.
Detailed description of the algorithm is relegated to the Appendix,
along with most of the technical proofs.

\section{A Review of Distance Covariance}
\label{sec:reviewDistCorr}

Distance covariance and distance correlation was introduced in 2005 by one of the co-authors of this paper, G. J. Sz\'ekely, in several lectures to address
the deficiency of Pearson's correlation, namely that the Pearson's correlation can be zero for dependent variables.
In the following, we start with a definition of the distance covariance.
\begin{defn}
The
population distance covariance between random vectors $X$ and $Y$ with finite first moments is the nonnegative number $\mathcal{V}(X, Y)$ defined by
a weighted $L_2$ norm measuring the distance between the joint characteristic
function (c.f.) $\phi_{X,Y}$ of $X$ and $Y$, and the product $\phi_X \phi_Y$
of the marginal
c.f.'s of $X$ and $Y$.
If $X$ and $Y$ take values in $\mathbb{R}^p$ and $\mathbb{R}^q$, respectively,
$\mathcal{V}^2(X, Y)$ is
\begin{eqnarray*}
\mathcal{V}^2(X, Y) &=& \|\phi_{X,Y}(t,s) - \phi_X(t) \phi_Y(s)\|_w^2 \\
&:=& \int_{\mathbb{R}^{p+q}} |\phi_{X,Y}(t,s) - \phi_X(t) \phi_Y(s)|^2
w(t,s) dt ds,
\end{eqnarray*}
where $w(t,s) := (|t|^{1+p}_{p} |s|^{1+q}_{q})^{-1}$.
The integral exists provided that $X$ and $Y$ have finite first moments.
\end{defn}

This immediately shows that distance covariance is zero if and only if
the underlying random variables are independent. The beauty of kernel $w(t,s) := (|t|^{1+p}_{p} |s|^{1+q}_{q})^{-1}$ is that the corresponding sample statistic has the following surprisingly simple form.
Denote the pairwise distances of the $X$ observations by $a_{ij}:= \|X_i - X_j\|$ and
the pairwise distances of the $Y$ observations by $b_{ij}:= \|Y_i- Y_j\|$ for $i,j =1, \dots, n$ and denote the corresponding double centered distance matrices by
$(A_{ij})_{i,j = 1}^n$, and $(B_{ij})_{i,j = 1}^n$ where
\begin{equation}\label{eq:Aij-0}
A_{ij}=\left\{ \begin{array}{ll}
a_{ij} -\frac{1}{n}\sum^{n}_{\ell=1} a_{i\ell} -\frac{1}{n}\sum^{n}_{k=1} a_{k j} +\frac{1}{n^2} \sum^{n}_{k,\ell=1} a_{k \ell}, & i \neq j; \\
0, & i=j.
\end{array}
\right.
\end{equation}

\begin{equation}\label{eq:Bij}
B_{ij}=\left\{ \begin{array}{ll}
b_{ij} -\frac{1}{n}\sum^{n}_{\ell=1} b_{i\ell} -\frac{1}{n}\sum^{n}_{k=1} b_{k j} +\frac{1}{n^2} \sum^{n}_{k,\ell=1} b_{k \ell}, & i \neq j; \\
0, & i=j.
\end{array}
\right.
\end{equation}

It is clear that the row sums and column sums of these double centered matrices are 0. The squared sample distance covariance is the following simple formula
$$ \frac{1}{n^2} \sum_{i,j=1}^n   A_{ij} B_{ij}.$$

The corresponding squared sample variance is
$$   \mathcal{V}^2_n (X) :=  \frac{1}{n^2} \sum_{i,j=1}^n   A_{ij}^2 $$
and we can define the sample distance correlation as the standardized sample covariance:

\begin{equation}\label{e:alphan}
  \mathcal{R}^2_n({X,Y})=
 \left\{
   \begin{array}{ll}
 \frac {\mathcal{V}^2_n({X,Y})}{\sqrt {
 \mathcal{V}^2_n (X) \mathcal{V}^2_n(Y)}} \;,
   & \hbox{$ \mathcal{V}^2_n
 (X)  \mathcal{V}^2_n (Y) >0
$;} \\
     0, & \hbox{$\mathcal{V}^2_n
 (X)  \mathcal{V}^2_n (Y) =0$.}
   \end{array}
 \right.
 \end{equation}

For more details see \citeasnoun{Szekely2007} and a discussion paper \cite{Szekely2009}.
It is clear that   $\mathcal{R}_n({X,Y})$ is rigid motion invariant and scale invariant.
%According to \cite{Szekely2012} this is essentially characteristic for the weight function $w(t,s) := (|t|^{1+p}_{p} |s|^{1+q}_{q})^{-1}$.
For recent applications of distance correlation, see e.g.,
\citeasnoun**{Li2012} and \citeasnoun{Kong2012}.

The population version of distance covariance and distance correlation can be defined without characteristic functions, see \citeasnoun{Lyons2013}. This definition is as follows.
Let $X \in \mathbb R^p$ and $Y \in \mathbb R^q$ be random variables with
finite expectations. The random distance functions are $ a(X, X'):= |X-X'|_p $ and $ b(Y,Y')=|Y- Y'|_q $.
Here the primed random variable $X'$ denotes an independent and identically distributed (i.i.d.) copy of the
variable $X$, and similarly $Y, Y'$ are i.i.d.

Introduce the real-valued function $$ m(x, F_X) = E[a(x, X)] = E|x - X| = \int |x - x'| dF_X(x'), $$ where
$F_X$ is the cumulative distribution function (cdf) of $X$,
and
$$
m(X, F_X) = \int |X - x'| dF_X(x'),
$$
which is a real-valued random
variable. For simplicity we write $m(x):=m(x, F_X)$ and $m(X):=m(X, F_X)$.

Next we define the counterpart of centered distance matrices. The centered distance function is
$$ a(x,x')
:=  a(x,x') - m(x) - m(x') + E[m(X')]. $$ For random variables we have
\begin{align*}
 A(X, X') &= a(X, X') - m(X) - m(X') + E[m(X')],
\end{align*}
where $$ E[m(X')] = E[m(X, F_X)] = \int \int |x - x'| \, dF_X(x') \, dF_X(x). $$

Similarly define the centered distance function $b(y, y')$ and the random variable $ B(Y, Y')$. Now for
$ X, X'$ i.i.d., and $Y, Y'$ i.i.d., such that $X$ and $Y$ have finite expectations, the population distance covariance
$\mathcal V(X,Y)$ is defined by
\begin{equation}\label{e:dcov2}
\mathcal V^2(X,Y):= E[ A(X,X')\,  B(Y,Y')].
\end{equation}
We have that $\mathcal V^2(X, Y)$ is always nonnegative, and equates zero if and only if $X$ and $Y$ are
independent.

It is clear by inspection that without further efforts the implementation of the sample distance covariance and the corresponding sample distance correlation requires O($n^2$) steps.
In this paper we show that for real-valued random variables $X$ and $Y$, we do not need more than O($n \log n$) steps.

\section{The unbiased version of the squared sample distance covariance: reformulation and relation to U-statistics}
\label{sec:unbiased}

In this section, a reformulation is given in Section \ref{sec:reformulation}.
We then show in Section \ref{sec:u-stat} that the newly formed unbiased estimator is a U-statistic.

\subsection{Reformulation}
\label{sec:reformulation}

We will work with the unbiased version of the squared sample distance covariance, which is published in \citeasnoun{pdc2014}. The definition is as follows.

\begin{defn}[$\mathcal{U}$-centered matrix]
Let $A = (a_{ij})$ be a symmetric, real valued $n \times n$ matrix with zero
diagonal, $n > 2$. Define the $\mathcal{U}$-centered matrix $\widetilde{A}$
as follows: the $(i,j)$-th entry of $\widetilde{A}$ is
\begin{equation}\label{eq:Aij}
\widetilde{A}_{ij}=\left\{ \begin{array}{ll}
a_{ij} -\frac{1}{n-2}\sum^{n}_{\ell=1} a_{i\ell} -\frac{1}{n-2}\sum^{n}_{k=1} a_{k j} +\frac{1}{(n-1)(n-2)} \sum^{n}_{k,\ell=1} a_{k \ell}, & i \neq j; \\
0, & i=j.
\end{array}
\right.
\end{equation}
\end{defn}
Here the ``$\mathcal{U}$-centered" is so named because as shown below, the corresponding inner
product (which will be specified in \eqref{eq:tAdottB})
defines an unbiased estimator of the squared distance covariance.
\begin{prop}
Let $(x_i, y_i), i = 1,\ldots, n$ denote a sample of observations from
the joint distribution $(X, Y)$ of random vectors $X$ and $Y$.
Let $A = (a_{ij})$ be the
Euclidean distance matrix of the sample $x_1,\ldots, x_n$ from the distribution of $X$,
and $B = (b_{ij})$ be the Euclidean distance matrix of the sample $y_1,\ldots, y_n$ from the
distribution of $Y$.
Then if $E(|X| + |Y|) < 1$, for $n > 3$, the following
\begin{equation} \label{eq:tAdottB}
(\widetilde{A} \cdot \widetilde{B}):= \frac{1}{n(n-3)} \sum_{i \neq j}
\widetilde{A}_{i j} \widetilde{B}_{i j}
\end{equation}
is an unbiased estimator of squared population distance covariance $\mathcal{V}^2(X, Y)$.
\end{prop}
The proof of the above proposition is in the appendix of \citeasnoun{pdc2014}.

Let $\Omega_n$ denote the inner product defined in \eqref{eq:tAdottB}.
The following notations will be used.
Define the column and row sums as follows:
\begin{eqnarray*}
a_{i\cdot}  = \sum^{n}_{\ell=1} a_{i \ell}, &&
a_{\cdot j} = \sum^{n}_{k=1} a_{k j}, \\
b_{i\cdot}  = \sum^{n}_{\ell=1} b_{i \ell}, &&
b_{\cdot j} = \sum^{n}_{k=1} b_{k j}, \\
a_{\cdot \cdot} = \sum^{n}_{k,\ell=1} a_{k \ell}, && \mbox{ and }
b_{\cdot \cdot} = \sum^{n}_{k,\ell=1} b_{k \ell}.
\end{eqnarray*}

We will need the following lemma.

\begin{lem} \label{lem:Omega1}
If $\Omega_n$ is the inner product  defined in \eqref{eq:tAdottB} then
we have
\begin{equation}\label{eq:Omega_n}
\Omega_n =
\frac{1}{n(n-3)} \sum_{i \neq j} a_{ij}b_{ij}
-\frac{2}{n(n-2)(n-3)} \sum_{i=1}^n a_{i\cdot}b_{i\cdot}
+ \frac{a_{\cdot \cdot} b_{\cdot \cdot}}{n(n-1)(n-2)(n-3)}.
\end{equation}
\end{lem}
For the proof see the Appendix.
Formula \eqref{eq:Omega_n} will be used to prove that
(i) the estimator in \eqref{eq:tAdottB} is a U-statistic and thus we can apply the relevant limit theorems to study its asymptotic behavior;
(ii) the estimator in \eqref{eq:tAdottB} can be computed in $O( n \log n)$ steps.

\subsection{Validating the Statistic is a U-Statistic}
\label{sec:u-stat}

Suppose $x_1, \ldots, x_n$ is a sample.
For positive integer $r$, let $\Phi_r$ denote all the distinct $r$-subsets of $\{1,2,\ldots,n\}$.
For a set $\varphi \subset \{1,\ldots,n\}$, we define notation $x_\varphi = \{x_i \mid i \in \phi\}$.
Let $h: \mathbb{R}^r \to R$ be a symmetric real-valued or complex-valued kernel function of $r$ variables.
For each $n\ge r$, the associated U-statistic of order $r$, $U_{nr}\colon \mathbb{R}^n \to R$, is equal to the average over ordered samples of size $r$ of the sample values $h(x_\varphi)$.
In other words,
\begin{equation}\label{eq:U-stat}
U_{nr}(x_1,\ldots, x_n) = \frac{1}{{n \choose r}} \sum_{\varphi \in \Phi_r} h(x_{\varphi}),
\end{equation}

For U-statistics, we can verify the following lemma.
\begin{lem}\label{lem:check01}
For $1\le i \le n$, we denote
\begin{equation}\label{eq:check02}
U^{-i}_{nr}(x_1,\ldots, x_n) = U_{n-1,r}(x_1,\ldots, x_{i-1},x_{i+1},\ldots, x_n),
\end{equation}
where $U_{n-1,r}(x_1,\ldots, x_{i-1},x_{i+1},\ldots, x_n)$ is defined in \eqref{eq:U-stat} after removing the element $x_i$.
Then we must have
\begin{equation}\label{eq:check01}
(n-r){n \choose r} U_{nr}(x_1,\ldots,x_n)
= \sum_{i=1}^n {n-1 \choose r} U^{-i}_{nr}(x_1,\ldots, x_n) .
\end{equation}
\end{lem}
\begin{proof}[Proof of Lemma \ref{lem:check01}]
In \eqref{eq:check01}, each term $h(x_\varphi)$ is counted $(n-r)$ times on both sides.
Hence the equality holds.
\end{proof}

In fact, using arithmetic deduction, one can prove that the converse of the above is also true. In other words, the {\it jackknife invariance} is a necessary and sufficient condition for being U-statistics.
For a very similar (equivalent) approach see \citeasnoun{Lenth1983}.
\begin{lem}\label{lem:check02}
If there exists a positive integer $r > 0$, such that for any $n > r$,
function $U_{nr}(x_1,\ldots,x_n)$ satisfies \eqref{eq:check01} and \eqref{eq:check02},
then there must be a kernel function $h(\cdot)$ of order $r$, such that
$U_{nr}(x_1,\ldots,x_n)$  can be written in a form as in \eqref{eq:U-stat}; i.e., $U_{nr}(x_1,\ldots,x_n)$ is a U-statistic.
\end{lem}
A proof of the above can be found in the Appendix.

The  two lemmas above show that the recursive relation \eqref{eq:check01} is a necessary and sufficient condition for a U-statistic.
For later use, we explicitly restate the result below.
\begin{theorem}\label{th:check01}
Let $\Omega_n(x_1,\ldots,x_n)$ be a statistic of a sample $x_1,\ldots,x_n$.
Let $\Omega^{-i}_{n-1}(x_1,\ldots,x_n)$, $i=1,2,\ldots,n$,
be a statistic of a reduced sample $x_1,\ldots,x_{i-1},x_{i+1}, \ldots,x_n$; i.e., $\Omega^{-i}_{n-1}(x_1,\ldots,x_n)$ is the statistic after removing the observation $x_i$.
The necessary and sufficient condition for $\Omega_n(x_1,\ldots,x_n)$ to be a U-statistic of order $r$ is
\begin{equation}\label{eq:check03}
n \cdot \Omega_n(x_1,\ldots,x_n) = \sum^n_{i=1} \Omega^{-i}_{n-1}(x_1,\ldots,x_n)
\end{equation}
holds for all $n \ge r$.
\end{theorem}
The above can be extended to a two-sample problem, in which a sample is
$(x_1,y_1), (x_2,y_2),\ldots, (x_n,y_n)$ for $n \ge 1$.
By replacing $x_i$ with $(x_i,y_i)$, all the previous arguments still hold.
\begin{proof}[Proof of Theorem \ref{th:check01}].
Combine Lemma \ref{lem:check01} and Lemma \ref{lem:check02}, and
simplify \eqref{eq:check01}, we have \eqref{eq:check03}.
\end{proof}

%\subsection{Back to Distance Covariance}

Let $\Omega_n$ denote the inner product that is defined in \eqref{eq:tAdottB}.
Note $\Omega_n$ is based on the entire sample (i.e., $(x_i,y_i), i=1,2,\ldots,
n$.
For $1 \le i \le n$, let $\Omega^{-i}_{n-1}$ denote the corresponding statistic
after knocking out pair $(x_i,y_i)$ from the entire sample.

The following lemma establish counterpart for $\Omega_{n-1}^{-k}$, where $k=1,2,\ldots,n$.
\begin{lem} \label{lem:Omega2}
For $1\le k \le n$, let $a_{i\cdot}^{-k}$, $b_{i\cdot}^{-k}$, $a_{\cdot \cdot}^{-k}$,
and $b_{\cdot \cdot}^{-k}$ denote the corresponding sums after entry $(x_k,y_k)$ is removed from
the sample.
If $\Omega_{n-1}^{-k}$ is the inner product that is defined in \eqref{eq:tAdottB}
after knocking off the $k$-th entry $(x_k,y_k)$,
we have
\begin{eqnarray}
\Omega_{n-1}^{-k} &=&
\frac{1}{(n-1)(n-4)} \sum_{i \neq j, i\neq k, j\neq k} a_{ij}b_{ij}
-\frac{2}{(n-1)(n-3)(n-4)} \sum_{i=1, i\neq k}^n a_{i\cdot}^{-k} b_{i\cdot}^{-k} \nonumber \\
&& +\frac{a_{\cdot \cdot}^{-k} b_{\cdot \cdot}^{-k}}{(n-1)(n-2)(n-3)(n-4)}. \label{eq:Omega_n-k}
\end{eqnarray}
\end{lem}
We will not provide the proof for Lemma \ref{lem:Omega2}, because it will be identical with
the proof of Lemma \ref{lem:Omega1}.

\begin{theorem}\label{th:u-stat}
Estimator $\Omega_n$---the inner product that is
defined in \eqref{eq:tAdottB}---is a U-statistic.
The kernel function of the corresponding U-statistic is
the inner product that was defined in
\eqref{eq:tAdottB} with $n=4$.
\end{theorem}
See a proof in the Appendix.

Now we know that \eqref{eq:tAdottB} is a U-statistic and it is easy to see that \eqref{eq:tAdottB} is in fact a U-statistic with a degenerate kernel under the null hypothesis of independence of $X$ and $Y$,
thus we can see from Corollary 4.4.2 of
\citeasnoun{KoroljukEtal_1994} that if the second moments of $X$ and $Y$ are finite then under the null hypothesis, the limit distribution of
$n (\widetilde{A} \cdot \widetilde{B})$  has the form
$\sum_{i=1}^{\infty} \lambda_i (Z_i^2 -1)$,
where  $\lambda_i \ge 0$, and $Z_i$ are i.i.d. standard normal random variables.
Under the alternative hypothesis we have that $n | (\widetilde{A} \cdot \widetilde{B})| \to \infty$, thus we can easily construct a consistent test of independence.
For a technically much more difficult approach, see \citeasnoun{Szekely2007} where a similar result was derived for a related V-statistic using deep results on complex-valued Gaussian processes.

\section{Fast Algorithm}
\label{sec:algorithm}

We now argue that when $X$ and $Y$ are univariate, there is an O$(n\log n)$ algorithm to implement \eqref{eq:Omega_n}.
We start with several intermediate results, which are presented as lemmas below.
\begin{lem} \label{lem:AiDot}
Denote
$$
x_\cdot = \sum^n_{i=1} x_i.
$$
For $1 \le i \le n$, we also denote
\begin{eqnarray*}
\alpha_i^x &=& \sum_{x_\ell < x_i} 1, \\
\beta_i^x  &=& \sum_{x_\ell < x_i} x_\ell.
\end{eqnarray*}
We have
\begin{equation}\label{eq:aidot}
a_{i \cdot} = x_\cdot + (2\alpha_i^x - n)x_i - 2 \beta_i^x.
\end{equation}
\end{lem}
A proof is relegated to the appendix.

Due to symmetry, the following is the counterpart fact for $Y$. We state it without a proof.
\begin{lem} \label{lem:BiDot}
Denote
$$
y_\cdot = \sum^n_{i=1} y_i.
$$
For $1 \le i \le n$, we denote
\begin{eqnarray*}
\alpha_i^y &=& \sum_{y_\ell < y_i} 1, \\
\beta_i^y  &=& \sum_{y_\ell < y_i} y_\ell.
\end{eqnarray*}
We have
\begin{equation}\label{eq:bidot}
b_{i \cdot} = y_\cdot + (2\alpha_i^y - n)y_i - 2 \beta_i^y.
\end{equation}
\end{lem}

Using formulas \eqref{eq:aidot} and \eqref{eq:bidot},
the following two equations can be easily established. We state them without a proof.
\begin{cor}\label{cor:abdots}
We have
\begin{equation}\label{eq:adotdot}
a_{\cdot \cdot} =  2 \sum_{i=1}^n \alpha_i^x x_i - 2 \sum_{i=1}^n \beta_i^x,
\end{equation}
and
\begin{equation}\label{eq:bdotdot}
b_{\cdot \cdot} =  2 \sum_{i=1}^n \alpha_i^y y_i - 2 \sum_{i=1}^n \beta_i^y.
\end{equation}
\end{cor}

The following lemma will be used.
\begin{lem}\label{lem:aijbij}
We define a sign function, for $\forall 1 \le i,j \le n$,
$$
S_{ij} = \left\{\begin{array}{ll}
+1, & \mbox{ if }(x_i-x_j)(y_i-y_j) > 0, \\
-1, & \mbox{ otherwise.}
\end{array}
 \right.
$$
For any sequence $\{c_j, j=1,\ldots, n\}$, for $1 \le i \le n$, we define
$$
\gamma_i(\{c_j\}) = \sum_{j: j \neq i} c_j S_{ij}.
$$
The following is true:
\begin{equation}\label{eq:sumab}
\sum_{i \neq j} a_{ij}b_{ij} = \sum_{i=1}^n \left[ x_i y_i \gamma_i(\{1\})
+ \gamma_i(\{x_j y_j\}) -x_i \gamma_i(\{y_j\}) - y_i \gamma_i(\{x_j\})  \right].
\end{equation}
\end{lem}

\begin{proof}[Proof of Lemma \ref{lem:aijbij}]
We have
\begin{eqnarray*}
\sum_{i \neq j} a_{ij}b_{ij} &=& \sum_{i \neq j} |x_i - x_j| \cdot |y_i - y_j| \\
&=& \sum_{i=1}^n \sum_{j: j \neq i} (x_i y_i + x_j y_j - x_i y_j -x_j y_i )S_{ij} \\
&=& \sum_{i=1}^n \left[x_i y_i \sum_{j: j \neq i} S_{ij} + \sum_{j: j \neq i} x_j y_j S_{ij}
- x_i \sum_{j: j \neq i} y_j S_{ij} - y_i \sum_{j: j \neq i} x_j S_{ij} \right].
\end{eqnarray*}
Per the definition of $\gamma_i(\{\cdots\})$, one can verify that the above equates to \eqref{eq:sumab}.
\end{proof}

\begin{lem}\label{lem:gamma}
For any sequence $\{c_j, j=1,\ldots, n\}$, there is an O$(n \log n)$ algorithm to compute for all
$\gamma_i(\{c_j\})$ ($= \sum_{j: j \neq i} c_j S_{ij}$),
where $i=1,\ldots, n$.
\end{lem}
Again, we relegate the proof to the appendix.
The main idea of the proposed algorithm is a modification as well as an extension of
the idea that was used in \citeasnoun{Knight1966} and \citeasnoun{Christensen2005}, which developed a fast algorithm
for computing the Kendall's $\tau$ rank correlation coefficient.
The principle of the AVL tree structure \cite{AVL1962} was adopted.
Despite they are in a similar spirit, the algorithmic details are different.
We now present the main result in the following theorem.
\begin{theorem}\label{th:fastAlgo}
The unbiased estimator of the squared population distance covariance
(that was defined in \eqref{eq:tAdottB}) can be computed by an
O($n \log n$) algorithm.
\end{theorem}
\begin{proof}[Proof of Theorem \ref{th:fastAlgo}]
In Lemma \ref{lem:Omega1},
the unbiased statistic has been rewritten as in \eqref{eq:Omega_n}.
For the first term on the right hand side of \eqref{eq:Omega_n},
per Lemmas \ref{lem:aijbij} and \ref{lem:gamma},
there is an O($n \log n$) algorithm to compute it.

For the second term on the right hand side of \eqref{eq:Omega_n},
Note that quantities $\alpha_i^x, \beta_i^x, \alpha_i^y$, and
$\beta_i^y$  that were defined in Lemmas \ref{lem:AiDot} and
\ref{lem:BiDot}, respectively, are partial sums, which can be
computed for all $i$'s with O($n \log n$) algorithms.
The $\log n$ factor is inserted, because one may need to sort
$x_i$'s or $y_i$'s in order to compute for
$\alpha_i^x, \beta_i^x, \alpha_i^y$, and
$\beta_i^y$.
Then by \eqref{eq:aidot} and \eqref{eq:bidot}, all
$a_{i \cdot}$ and $b_{i \cdot}$ can be computed at order O($n \log n$).
Consequently, the second term on the right hand side of \eqref{eq:Omega_n}
can be computed by using an O($n \log n$) algorithm.

For the third term on the right hand side of \eqref{eq:Omega_n},
using \eqref{eq:adotdot} and \eqref{eq:bdotdot}
in Corollary \ref{cor:abdots},
we can easily see that it can be computed via an O($n \log n$) algorithm.
From all the above, the theorem is established.
\end{proof}

For readers' convenience, we present a detailed algorithm description in Appendix, where Algorithm \ref{alg:PartialSum2D} realizes the idea that is described in the proof
of Lemma \ref{lem:gamma};
Algorithm \ref{alg:sub01} is a subroutine that will be called in Algorithm \ref{alg:PartialSum2D}; and
the Algorithm \ref{alg:FaDCor} is the algorithm that can compute for
the distance covariance at O($n \log n$).

\section{Numerical Experiments}
\label{sec:simulations}

In Section \ref{sec:implement}, we describe a MATLAB and C based implementation of the
newly proposed fast algorithm.
This fast algorithm enables us to run some simulations with sample sizes
that were impossible to experiment with before its appearance.
We report some numerical experiments in Section \ref{sec:effectiveness}.
Distance correlation has been found helpful in feature screening.
In Section \ref{sec:screening}, we redo experiments on this regard,
increasing the sample size from $n=200$ to $n=20,000$.
It is observed that the advantage of using the distance correlation is
more evident when the sample size becomes larger.

\subsection{Matlab Implementation}
\label{sec:implement}

The fast algorithm was implemented in MATLAB, with a key step (of dyadic updating) being implemented
in C.
It was then compared against the direct (i.e., slow) implementation.
Table \ref{tab:faDCor01} presents the average running time for the two different implementations in MATLAB with $1,000$ replications at each sample size.
The sample size goes from $32$ ($=2^5$) to $2048$ ($=2^{11}$).
In all these cases, the two methods ended with identical solutions; this validates our fast algorithm.
Note a comparison in MATLAB is not desirable for our fast algorithm.
The direct method calls some MATLAB functions, which achieve the speed of a low-level language implementation, while the implementation of the fast method is not.
In theory, the fast algorithm will compare more favorably if both methods are implemented in a low-level language, such as in C or C++.
\begin{table}[htbp]
\begin{center}
\begin{tabular}{r|rr}
\hline
Sample Size & Direct method & Fast method \\
\hline
32	& 0.0006 (0.0001) & 0.0014 (0.0001) \\
64	& 0.0008 (0.0001) & 0.0024 (0.0002) \\
128	& 0.0019 (0.0004) & 0.0053 (0.0006) \\
256	& 0.0083 (0.0010) & 0.0120 (0.0011) \\
512	& 0.0308 (0.0021) & 0.0272 (0.0018) \\
1024	& 0.1223 (0.0051) & 0.0647 (0.0037) \\
2048	& 0.4675 (0.0172) & 0.1478 (0.0045) \\
\hline
\end{tabular}
\end{center}
\caption{Running times (in seconds)
for the direct method and the fast method for computing the distance correlations.
The values in the parentheses are sample standard errors.
At each sample size, $1,000$ repetitions were run.\label{tab:faDCor01}}
\end{table}
\begin{figure}[htbp]
  \centering
  \includegraphics[width = 0.8\linewidth]{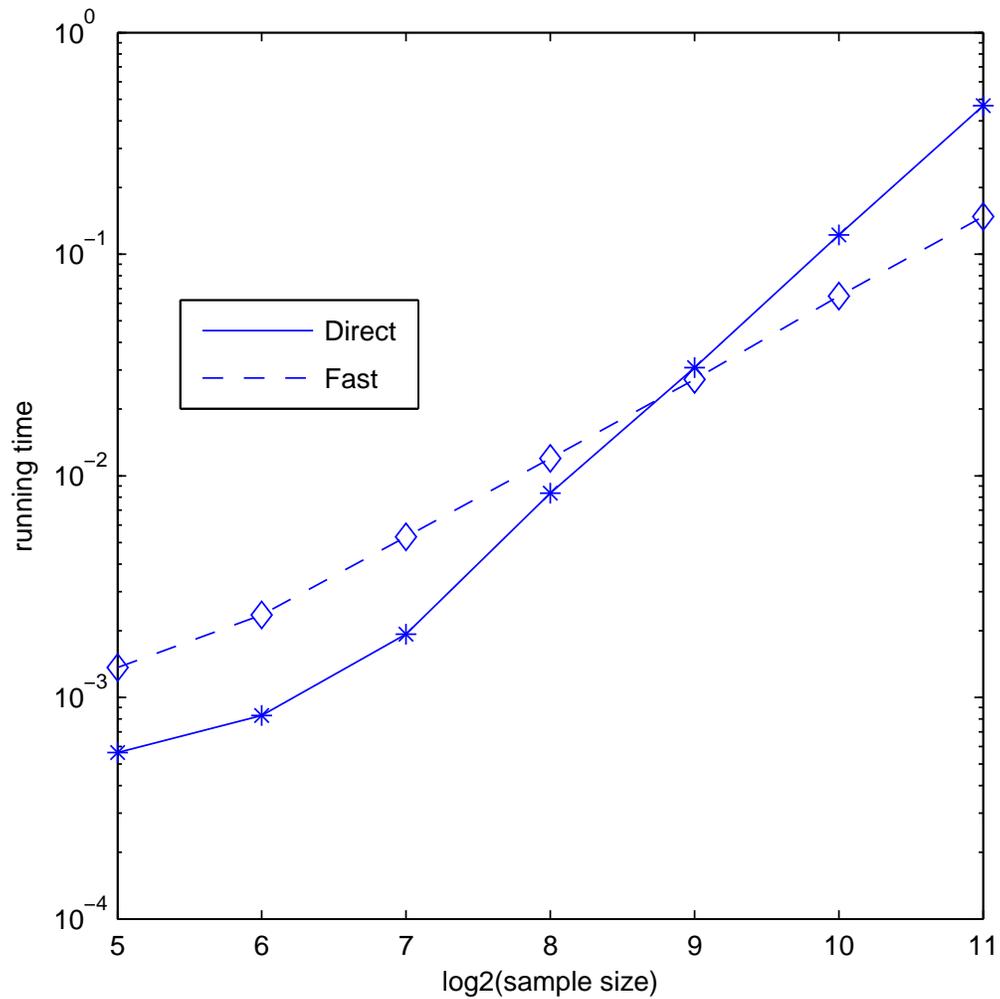}
  \caption{A comparison of running time between the direct method and the fast method for the computation of the distance correlations. }\label{fig:dcor01}
\end{figure}
Fig. \ref{fig:dcor01} provides a visual comparison of the two methods.
All the experiments that are reported in this paper is run on a laptop (Lenovo T520, Intel Core i7-2640M CPU @ 2.80GHz) with allowable 975 MB memory in MATLAB Version 8.2.0.89 (R2013b).

When the sample size is large, e.g., when $n=4096$,
the direct method will generate an ``out-of-memory'' message.
Recall the direct method computes for all pairwise distances, hence it requires O$(n^2)$ memory.
The fast method only requires O$(n)$ in memory.
For illustration purpose, we run the fast algorithm for sample size $n$ going from $4,096$ (which
is $2^{12}$) to $1,048,576$ (which is $2^{20}$).
The running times are reported in Tab. \ref{tab:faDCor02} and Fig. \ref{fig:dcor02}.
When $n=1,048,576$, the running time is a little more than three minutes.
The trend that is observable from Fig. \ref{fig:dcor02} consists with our claim that the
fast method is an O$(n \log n)$ algorithm.
It is evident that the running time scales approximately linearly with the sample size ($n$).
We did not run experiments with larger sample sizes, because their outcomes are predictable
by property of the fast method.
\begin{table}[htbp]
\begin{center}
\begin{tabular}{rr}
\hline
Sample Size &  Fast method \\
\hline
4,096	  &    0.3323 (0.0044) \\
8,192	  &    0.7432 (0.0051) \\
16,384	  &    1.6752 (0.0137) \\
32,768	  &    3.7686 (0.0238) \\
65,536	  &    8.5158 (0.0654) \\
131,072	  &   19.1241 (0.4688) \\
262,144	  &   42.2150 (0.3918) \\
524,288	  &   93.1250 (0.6422) \\
1,048,576 &  204.3403 (1.7328) \\
\hline
\end{tabular}
\end{center}
\caption{Running times (in seconds)
for the fast method for computing the distance correlations, when the sample sizes are large.
The values in the parentheses are sample standard errors.
At each sample size, $100$ repetitions were run.\label{tab:faDCor02}}
\end{table}
\begin{figure}[htbp]
  \centering
  \includegraphics[width = 0.8\linewidth]{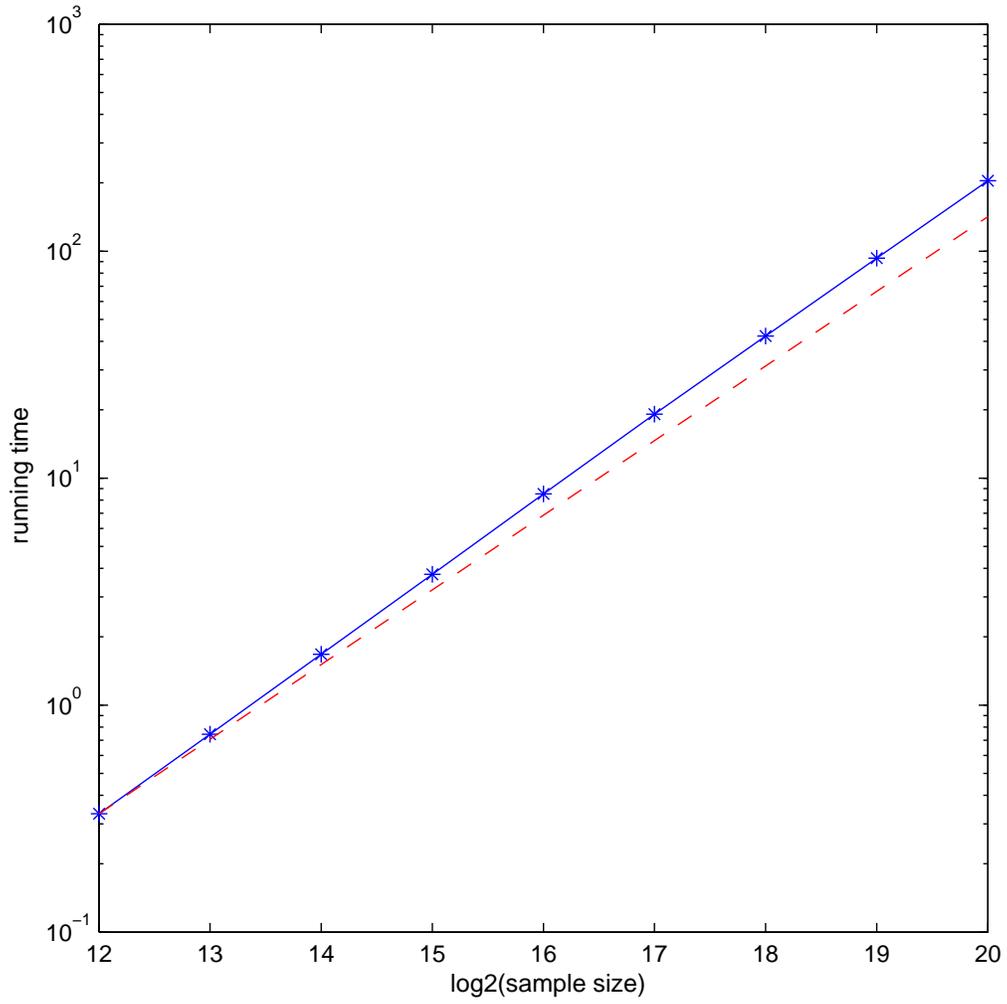}
  \caption{An illustration of running times of the fast method for the computation of the distance correlations.
The dashed line corresponds to an O$(n \log n)$ algorithm. }\label{fig:dcor02}
\end{figure}

\subsection{Measuring Effectiveness of Distance Correlation}
\label{sec:effectiveness}

The distance correlation is zero if and only if the corresponding two random variables are independent of each other.
The Pearson's correlation does not have such a property.
There have been intuitive numerical examples to illustrate such an advantage of using the
distance correlation.
See the Wikipedia page on ``distance correlation."
When the direct implementation of the distance correlation is adopted, the sample size
(which was denoted by $n$) cannot
be large, due to the O$(n^2)$ complexity of the direct method.
In Fig. \ref{fig:showcases01}, we compare the Pearson's correlation with the distance correlation
in nine representative cases:
\begin{enumerate}
\item[(1)] $(X,Y)$ is a bivariate normal with moderate correlation;
\item[(2)] a bivariate normal with a correlation close to $1$;
\item[(3)] a thickened rippled curve;
\item[(4)] a rotation of a uniformly filled rectangle;
\item[(5)] a further rotation of the aforementioned uniformly filled rectangle;
\item[(6)] a thickened quadratic curve;
\item[(7)] bifurcated quadratic curves;
\item[(8)] a thickened circle; and
\item[(9)] a bivariate mixed normal with independent coordinates.
\end{enumerate}
When the sample sizes are $40$ and $400$, respectively, Fig. \ref{fig:showcases01}
presents the Pearson's correlation and the distance correlation in all cases.
In the cases (3) through (8), we seemingly observe the trend that the Pearson's correlations are getting
close to zero, while the distance correlations are not.
However, the significance of such a pattern is not evident.
\begin{figure}[htbp]
  \centering
  \begin{tabular}{c}
  (a) Sample size: $n=40$ \\
  \includegraphics[scale=1.0]{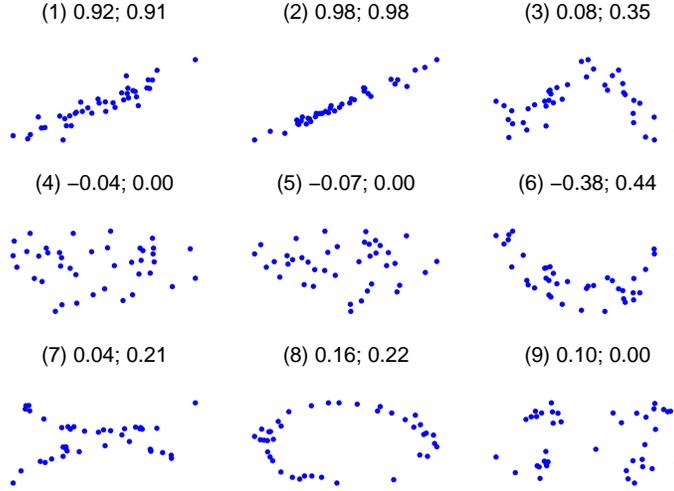} \\
  (b) $n=400$ \\
  \includegraphics[scale=1.0]{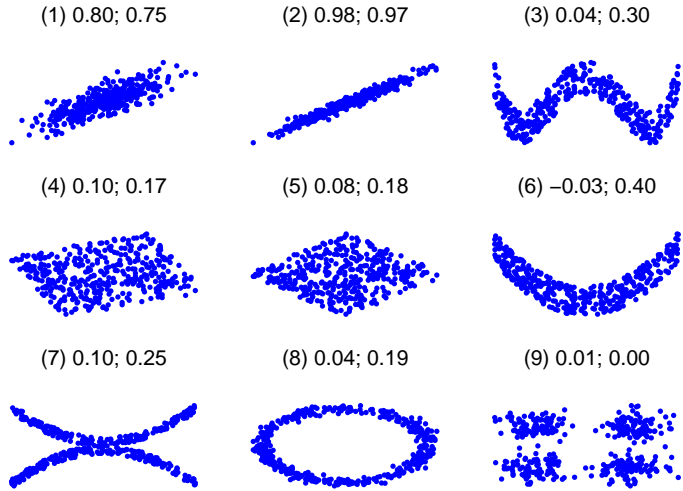}
  \end{tabular}
  \caption{Comparison of the Pearson's correlation and the distance correlation in nine cases.
  In each sub-figure, the two coordinates correspond to the random variables $(X,Y)$.
  Each dot is a sample point. In the title, the first value is the Pearson's correlation, and the second one is the corresponding distance correlation. }\label{fig:showcases01}
\end{figure}

With the fast method, we now can run the same experiments with larger sample sizes.
In Fig. \ref{fig:showcases02}, we run the comparison with sample size $n=10,000$.
\begin{figure}[htbp]
  \centering
  $n=10,000$ \\
  \includegraphics[scale=1.0]{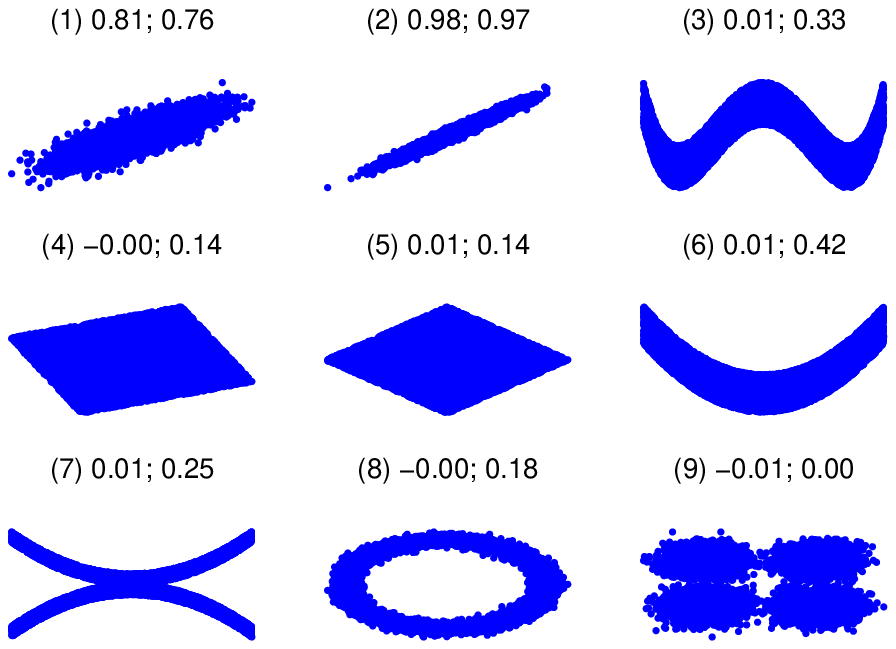}
  \caption{Comparison of the Pearson's correlation and the distance correlation when the sample size ($n$) is large: $n=10,000$. Each dot is a realization of a pair of random variables ($X,Y$).
  In each case, the first value is the Pearson's correlation, and the second one is the corresponding distance correlation.
  We can clearly observe that in the cases (3) through (8), the Pearson's correlations are close to zero, while the distance correlations are not. In cases of (1), (2), and (9), two correlations are close to each other, as the theory predicts. }\label{fig:showcases02}
\end{figure}
This is a sample size for which the corresponding experiment cannot be done with the direct method.
It is clear that the Pearson's correlation become nearly zero in the cases of (3) through (8),
even though the two random variables are not independent.
The corresponding distance correlations are clearly far from zero.
For readers' convenience, we summarize the results in Table \ref{tab:show03}.
\begin{table}[htbp]
\centering
\begin{tabular}{l|rrr rrr}
Sample Size & (1) &  (2) &  (3) &  (4) &  (5) &  (6) \\
           \hline
n = 40     & 0.92; 0.91 & 0.98; 0.98 & 0.08; 0.35 & -0.04; 0.00 & -0.07; 0.00 & -0.38; 0.44 \\
n = 400    & 0.80; 0.75 & 0.98; 0.97 & 0.04; 0.30 & 0.10; 0.17  & 0.08; 0.18  & -0.03; 0.40 \\
n = 10,000 & 0.81; 0.76 & 0.98; 0.97 & 0.01; 0.33 & 0.00; 0.14 & 0.01; 0.14  & 0.01; 0.42  \\
\hline
           &  (7) &  (8) &  (9)&&& \\
n = 40     & 0.04; 0.21 & 0.16; 0.22 & 0.10; 0.00 &&& \\
n = 400    & 0.10; 0.25 & 0.04; 0.19 & 0.01; 0.00 &&& \\
n = 10,000 & 0.01; 0.25 & 0.00; 0.18 & -0.01; 0.00 &&&
\end{tabular}
\caption{Pearson's correlations (left) and distance correlations (right) for the nine cases
 that are studied in Figures \ref{fig:showcases01} and \ref{fig:showcases02}.
It is of particular interests to observe that when $n=10,000$, the distance correlations in cases
(3) through (8) are clearly nonzero, while the Pearson's correlations in these cases converge
to zero.} \label{tab:show03}
\end{table}

The fast method allows us to study how the sample distance correlation converge to the population counterpart as a function of the sample size.
Fig. \ref{fig:converge01} shows the convergence of the sample distance correlation and Pearson's
correlation.
It is worth noting that in cases (3)-(8), the Pearson's correlation quickly converges to zero,
while the sample distance correlation clearly stays away from zero.
This experiments shows that a previous observation in Fig. \ref{fig:showcases02} should occur with
large probability.
\begin{figure}[htbp]
  \centering
  \includegraphics[scale=1.0]{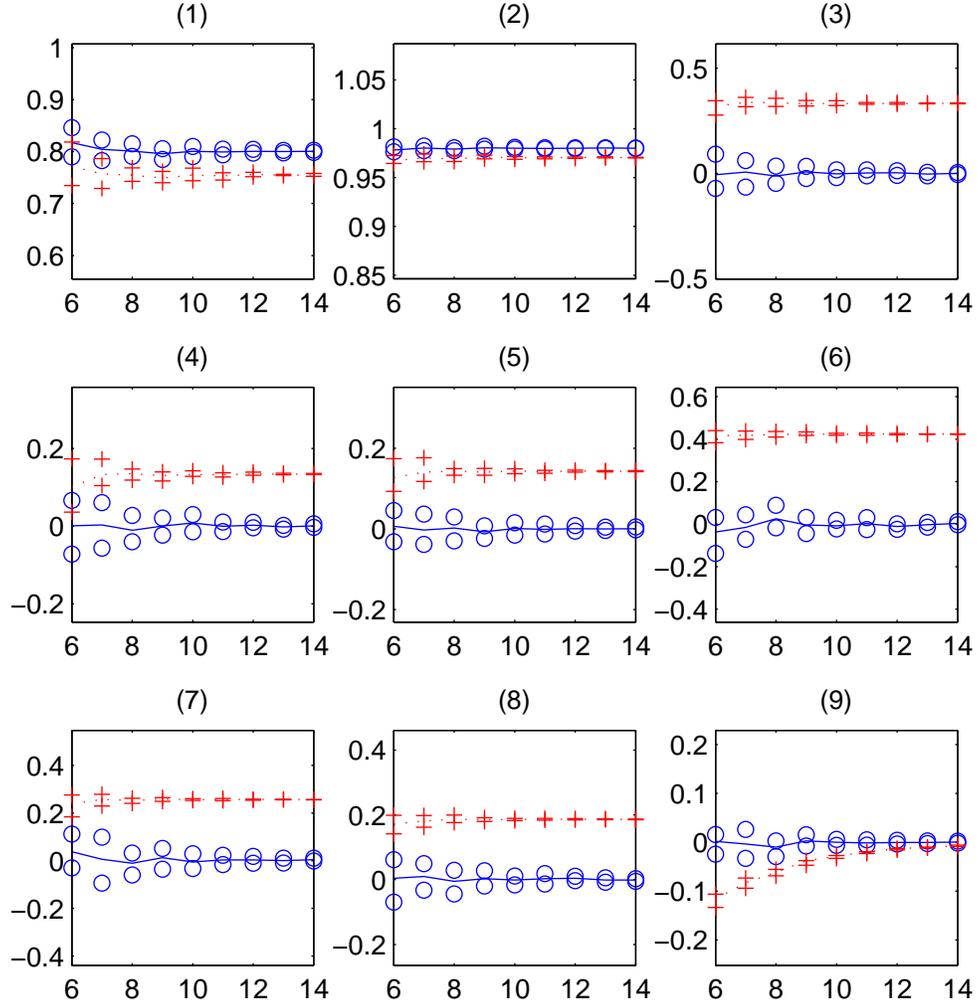}
  \caption{Convergence of $50\%$ covering interval of both sample Pearson's correlation
  (solid line, with low and upper sample quartiles marked by `$\circ$') and sample distance
  correlation (dotted lines, with both quartiles marked by `$+$').
  The horizontal axis equates the $\log_2($sample size$)$.
  The vertical axis corresponds to the values of correlations.
  In cases (3) through (8),
  the two correlations clearly converge to different constants, when the Pearson's correlation
  always seems to converge to zero.  }\label{fig:converge01}
\end{figure}

\subsection{Feature Screening}
\label{sec:screening}

In \citeasnoun**{Li2012}, distance correlation has been proposed to facilitate feature screening in ultrahigh-dimensional data analysis.
The proposed sure independence screening procedure based on the distance correlation (DC-SIS) has been proven to be effective in their simulation study.
Due to the use of the direct method, they restricted their sample size to
$n=200$.
We redo the simulations as in \citeasnoun**{Li2012}, however increases the sample
size to $n=20,000$, i.e., $100$ times of the originally attempted.
It is observed that the use
of distance correlation becomes more advantageous when the sample size increases.

The screening algorithm, which was initially advocated by \citeasnoun{jrssb2008SIS},
works as follows.
For each covariate $X_k, 1 \le k \le n$, a `marginal utility' function was
computed.
Such a marginal utility function can be the Pearson's correlation, the
distance correlation that was discussed in this paper, or other dependence
measure such as the one in \citeasnoun**{jasa2011SIR} that was also used in the
simulation studies of \citeasnoun**{Li2012}.
The `screening' is based on the magnitude of the values of these marginal
utility function.
Sometimes, forward, backward, or a hybrid stepwise approach is proposed.
In this paper, we refrain from further discussion in this potential
research direction.

Our simulation setup follows the one in \citeasnoun[Section 3]{Li2012}.
Note that an alternative approach named
{\it sure independent ranking and screening} (SIRS) \cite{jasa2011SIR}
was compared against.
For a sample, $(x_1,y_1), \ldots, (x_n,y_n)$, of two random variable
$X$ and $Y$, the SIRS dependence measure (i.e., the marginal utility
function) is defined as
\begin{equation}\label{eq:sirs01}
\mbox{SIRS}(X,Y) = \frac{1}{n(n-1)(n-2)}
 \sum_{j=1}^n \left[\sum_{i=1}^n x_i \mathbf{1}(y_i < y_j) \right]^2,
\end{equation}
where $\mathbf{1}(\cdot)$ is an indicator function.
The formulation in the above definition seemingly hint an O$(n^2)$
algorithm.
The following theorem will show that it can be computed via an
O$(n \log n)$ algorithm. The proof and the algorithmic
details are relegated to the appendix.
\begin{theorem}\label{th:sirs}
For a sample, $(x_1,y_1), \ldots, (x_n,y_n)$, of a bivariate random
vector $(X,Y)$, the SIRS measure \cite{jasa2011SIR} in \eqref{eq:sirs01}
can be computed via an algorithm whose average complexity is
O$(n \log n)$.
\end{theorem}

For completeness, we state our simulation setup below.
we generate $\mathbf{x} = (X_1,X_2, \ldots, X_p)^T$ from
normal distribution with zero mean and covariance matrix
$\Sigma = (\sigma_{ij})_{p \times p}$,
and the error term $\varepsilon$ from the standard normal distribution
$\mathcal{N}(0, 1)$.
Two covariance matrices are considered to assess
the performance of the DC-SIS and to compare with existing
methods: (1) $\sigma_{ij} = 0.8^{|i-j|}$ and (2) $\sigma_{ij} = 0.5^{|i-j|}$.
Note that a covariance matrix with entries $\sigma_{ij} = \rho^{|i-j|}, 0 < \rho < 1,$
enjoys a known Cholesky decomposition: $\Sigma  =R^T R$,
where $R =(r_{ij}) \in \mathbb{R}^{p\times p}, r_{ij}=0, $ if $j<i$, and $r_{1j}=\rho^{j-1}$,
$r_{ij}=c\cdot\rho^{j-i}$, for $i\ge 2$ and $j\ge i$, $c^2 + \rho^2=1$.
In our simulations, we take advantage of this known decomposition.
The dimension $p$ varies from $2000$ to $5000$.
Each experiment was repeated $500$ times, and the performance is evaluated
through the following three criteria:
\begin{enumerate}
\item $\mathcal{S}$: the minimum model size to include all active predictors.
We report the $5\%, 25\%, 50\%, 75\%$, and $95\%$
quantiles of $\mathcal{S}$ out of $500$ replications.

\item $\mathcal{P}_s$: the proportion that an individual active predictor is
selected for a given model size $d$ in the $500$ replications.

\item $\mathcal{P}_a$: the proportion that all active predictors are selected
for a given model size $d$ in the $500$ replications.

\end{enumerate}
The $\mathcal{S}$ is used to measure the model complexity of the resulting
model of an underlying screening procedure.
The closer to the
minimum model size the $\mathcal{S}$ is, the better the screening procedure
is.
The sure screening property ensures that $\mathcal{P}_s$ and $\mathcal{P}_a$
are both
close to one when the estimated model size $d$ is sufficiently large.
Different from  \citeasnoun{Li2012},
the $d$ is chosen to be $d_1 = [p/ 10 \log n]$, $d_2 = 2 d_1$, and $d_3 =
3 d_3$ throughout our simulations to empirically examine
the effect of the cutoff, where $[a]$ denotes the integer part of $a$.

An innovative stopping rule is introduced in \citeasnoun{Kong2012} for DC-SIS.
We did not implement it here, because the new stopping rule requires a multivariate version of the distance correlation, which is not covered
by this paper.

%{\it Example 1.}
The example is designed to compare the finite
sample performance of the DC-SIS with the SIS \cite{jrssb2008SIS} and the SIRS \cite{jasa2011SIR}.
In this example, we
generate the response from the following four models:
\begin{itemize}
\item[(1.a):] $Y = c_1 \beta_1 X_1 + c_2 \beta_2 X_2 + c_3 \beta_3
 \mathbf{1}(X_{12} < 0) + c_4 \beta_4 X_{22} + \varepsilon$,
\item[(1.b):] $Y = c_1 \beta_1 X_1 X_2 + c_3 \beta_2 \mathbf{1}(X_{12} < 0)
 + c_4 \beta_3 X_{22} + \varepsilon$,
\item[(1.c):] $Y = c_1 \beta_1 X_1 X_2
 + c_3 \beta_2 \mathbf{1}(X_{12} < 0) X_{22} + \varepsilon$,
\item[(1.d):] $Y = c_1 \beta_1 X_1 + c_2 \beta_2 X_2 + c_3 \beta_3
 \mathbf{1}(X_{12} < 0) + \exp(c_4|X_{22}|) \varepsilon$,
\end{itemize}
where $\mathbf{1}(X_{12} < 0)$ is an indicator function.

The regression functions $E(Y | x)$ in models (1.a)-(1.d)
are all nonlinear in $X_{12}$. In addition, models (1.b) and
(1.c) contain an interaction term $X_1 X_2$, and model (1.d)
is heteroscedastic. Following Fan and Lv (2008), we
choose $\beta_j = (-1)^U (a + |Z|)$ for $j = 1, 2, 3$, and $4$, where
$a = 4 \log n/\sqrt{n}$, $U\sim$Bernoulli$(0.4)$ and $Z \sim \mathcal{N}(0, 1)$.
We set
$(c_1, c_2, c_3, c_4) = (2, 0.5, 3, 2)$ in this example to be consistent with the
experiments in \citeasnoun{Li2012}: challenging the
feature screening procedures under consideration. For each
independence screening procedure, we compute the associated
marginal utility between each predictor $X_k$ and the response
$Y$. That is, we regard $\mathbf{x} = (X_1, \ldots, X_p)^T \in
\mathbb{R}^p$ as the predictor
vector in this example.

\begin{table}[htbp]
\centering
{\small
\begin{tabular}{clllll|lllll|lllll}
$\mathcal{S}$ & \multicolumn{5}{c}{SIS} & \multicolumn{5}{c}{SIRS} & \multicolumn{5}{c}{DC-SIS} \\
Model & $5\%$ & $25\%$ & $50\%$ & $75\%$ & $95\%$ & $5\%$ & $25\%$ & $50\%$ & $75\%$ &
$95\%$ & $5\%$ & $25\%$ & $50\%$ & $75\%$ & $95\%$ \\
 \cline{2-6}  \cline{7-11}  \cline{12-16}
 & \multicolumn{15}{c}{Case 1: $p=2000$ and $\sigma_{ij}=0.5^{|i-j|}$}  \\
(1.a) &    4&    4&    6&   10&   22&    4&    5&    6&   10&   20&    4&    5&    6&    9&   20 \\
(1.b) &   76&  551& 1180& 1592& 1918&  237&  814& 1269& 1789& 1959&    4&    6&    8&   11&   14 \\
(1.c) &  591&  922& 1364& 1781& 1941&  342&  827& 1354& 1637& 1930&    6&    6&    6&    8&   11 \\
(1.d) &    8&  237&  726& 1310& 1827&   58&  273&  919& 1444& 1878&    4&    4&    6&    8& 1001 \\
 & \multicolumn{15}{c}{Case 2: $p=2000$ and $\sigma_{ij}=0.8^{|i-j|}$}  \\
(1.a) &    5&    9&   14&   21&   46&    5&    9&   14&   22&   69&    4&    9&   14&   20&   36 \\
(1.b) &   28&   32&   35&  109& 1497&   29&   33&   40&  327& 1658&    4&   10&   15&   20&   26 \\
(1.c) &   39&  222&  711& 1418& 1924&   37&  109&  379& 1053& 1843&   10&   11&   14&   18&   23 \\
(1.d) &   13&   89&  547& 1152& 1823&   22&   77&  338&  863& 1679&    5&    8&   11&   17&  355 \\
 & \multicolumn{15}{c}{Case 3: $p=5000$ and $\sigma_{ij}=0.5^{|i-j|}$}  \\
(1.a) &    4&    5&    6&    9&   19&    4&    5&    6&    9&   20&    4&    5&    6&    9&   19 \\
(1.b) &   59& 1107& 2751& 3920& 4737&  299& 1755& 3255& 4289& 4837&    4&    6&    8&   10&   14 \\
(1.c) &  998& 2652& 3790& 4425& 4898&  321& 1864& 3269& 4303& 4857&    6&    6&    6&    8&   10 \\
(1.d) &   10&  221& 1346& 3055& 4585&   64&  596& 1894& 3500& 4791&    4&    4&    5&    7& 1024 \\
 & \multicolumn{15}{c}{Case 4: $p=5000$ and $\sigma_{ij}=0.8^{|i-j|}$}  \\
(1.a) &    5&   10&   16&   23&   46&    6&   11&   17&   23&   48&    5&   10&   16&   22&   35 \\
(1.b) &   28&   32&   36&  314& 3907&   29&   34&   49&  665& 4449&    5&    9&   14&   20&   27 \\
(1.c) &   45&  447& 1759& 3538& 4818&   41&  235& 1063& 2603& 4762&   10&   11&   14&   18&   23 \\
(1.d) &   14&  216& 1307& 3018& 4695&   23&  115&  747& 2135& 4368&    5&    8&   11&   16&  150
\end{tabular}
}
\caption{ The $5\%, 25\%, 50\%, 75\%$, and $95\%$ quantiles of the
minimum model size $\mathcal{S}$ out of $500$ replications.} \label{tab:ex1a}
\end{table}

\begin{table}[htbp]
\centering
{\footnotesize
\begin{tabular}{cclllll|lllll|lllll}
&& \multicolumn{5}{c}{SIS} & \multicolumn{5}{c}{SIRS} & \multicolumn{5}{c}{DC-SIS} \\
\cline{4-6} \cline{9-11} \cline{14-16}
&& \multicolumn{4}{c}{$\mathcal{P}_s$} & $\mathcal{P}_a$ &
\multicolumn{4}{c}{$\mathcal{P}_s$} & $\mathcal{P}_a$ &
\multicolumn{4}{c}{$\mathcal{P}_s$} & $\mathcal{P}_a$ \\
Model & Size & $X_1$ & $X_2$ & $X_{12}$ & $X_{22}$ & All
& $X_1$ & $X_2$ & $X_{12}$ & $X_{22}$ & All
& $X_1$ & $X_2$ & $X_{12}$ & $X_{22}$ & All \\
\hline
 && \multicolumn{15}{c}{Case 1: $p=2000$ and $\sigma_{ij}=0.5^{|i-j|}$}  \\
(1.a)&$d_1$& 1.00& 0.95& 1.00& 1.00& 0.95& 1.00& 0.95& 1.00& 1.00& 0.95& 1.00& 0.95& 1.00& 1.00& 0.95\\
     &$d_2$& 1.00& 0.96& 1.00& 1.00& 0.96& 1.00& 0.96& 1.00& 1.00& 0.96& 1.00& 0.96& 1.00& 1.00& 0.96\\
     &$d_3$& 1.00& 0.96& 1.00& 1.00& 0.96& 1.00& 0.96& 1.00& 1.00& 0.96& 1.00& 0.97& 1.00& 1.00& 0.97\\
(1.b)&$d_1$& 0.01& 0.02& 1.00& 1.00& 0.01& 0.00& 0.00& 1.00& 1.00& 0.00& 1.00& 1.00& 1.00& 1.00& 1.00\\
     &$d_2$& 0.05& 0.06& 1.00& 1.00& 0.02& 0.01& 0.03& 1.00& 1.00& 0.00& 1.00& 1.00& 1.00& 1.00& 1.00\\
     &$d_3$& 0.07& 0.08& 1.00& 1.00& 0.03& 0.02& 0.04& 1.00& 1.00& 0.00& 1.00& 1.00& 1.00& 1.00& 1.00\\
(1.c)&$d_1$& 0.08& 0.07& 0.00& 1.00& 0.00& 0.00& 0.03& 1.00& 1.00& 0.00& 1.00& 1.00& 1.00& 1.00& 1.00\\
     &$d_2$& 0.11& 0.07& 0.01& 1.00& 0.00& 0.01& 0.03& 1.00& 1.00& 0.00& 1.00& 1.00& 1.00& 1.00& 1.00\\
     &$d_3$& 0.14& 0.11& 0.01& 1.00& 0.00& 0.02& 0.04& 1.00& 1.00& 0.01& 1.00& 1.00& 1.00& 1.00& 1.00\\
(1.d)&$d_1$& 0.60& 0.40& 0.52& 0.52& 0.11& 0.98& 0.93& 1.00& 0.03& 0.03& 0.98& 0.89& 1.00& 1.00& 0.88\\
     &$d_2$& 0.72& 0.42& 0.54& 0.59& 0.12& 0.98& 0.95& 1.00& 0.04& 0.04& 0.98& 0.90& 1.00& 1.00& 0.89\\
     &$d_3$& 0.78& 0.44& 0.63& 0.62& 0.14& 0.98& 0.95& 1.00& 0.05& 0.05& 0.98& 0.91& 1.00& 1.00& 0.89\\
 && \multicolumn{15}{c}{Case 2: $p=2000$ and $\sigma_{ij}=0.8^{|i-j|}$}  \\
(1.a)&$d_1$& 0.92& 0.85& 0.88& 1.00& 0.74& 0.92& 0.84& 0.88& 1.00& 0.73& 0.92& 0.84& 0.91& 1.00& 0.76\\
     &$d_2$& 0.99& 0.97& 0.98& 1.00& 0.95& 0.99& 0.97& 0.98& 1.00& 0.94& 0.99& 0.97& 1.00& 1.00& 0.96\\
     &$d_3$& 0.99& 0.97& 0.99& 1.00& 0.95& 0.99& 0.97& 0.99& 1.00& 0.95& 0.99& 0.97& 1.00& 1.00& 0.96\\
(1.b)&$d_1$& 0.00& 0.01& 0.97& 1.00& 0.00& 0.00& 0.00& 0.97& 1.00& 0.00& 0.82& 0.84& 0.96& 1.00& 0.77\\
     &$d_2$& 0.64& 0.74& 0.99& 1.00& 0.64& 0.52& 0.65& 0.98& 1.00& 0.50& 1.00& 1.00& 1.00& 1.00& 1.00\\
     &$d_3$& 0.71& 0.79& 0.99& 1.00& 0.70& 0.62& 0.71& 0.99& 1.00& 0.61& 1.00& 1.00& 1.00& 1.00& 1.00\\
(1.c)&$d_1$& 0.01& 0.01& 0.50& 1.00& 0.00& 0.00& 0.00& 0.99& 1.00& 0.00& 0.94& 0.94& 0.93& 1.00& 0.87\\
     &$d_2$& 0.11& 0.09& 0.91& 1.00& 0.05& 0.11& 0.13& 1.00& 1.00& 0.07& 1.00& 1.00& 1.00& 1.00& 1.00\\
     &$d_3$& 0.15& 0.14& 0.91& 1.00& 0.09& 0.19& 0.25& 1.00& 1.00& 0.16& 1.00& 1.00& 1.00& 1.00& 1.00\\
(1.d)&$d_1$& 0.66& 0.60& 0.49& 0.42& 0.09& 1.00& 0.94& 0.97& 0.02& 0.02& 0.94& 0.86& 0.97& 1.00& 0.82\\
     &$d_2$& 0.72& 0.64& 0.55& 0.49& 0.16& 1.00& 0.96& 0.97& 0.18& 0.17& 0.99& 0.95& 0.98& 1.00& 0.91\\
     &$d_3$& 0.76& 0.68& 0.59& 0.54& 0.20& 1.00& 0.96& 0.98& 0.22& 0.20& 0.99& 0.95& 0.99& 1.00& 0.93
\end{tabular}
}
\caption{The proportions of $\mathcal{P}_s$ and $\mathcal{P}_a$ in our experiment for the first
two cases.
The user-specified model sizes are $d_1 = [p/ 10\log n ]$, $d_2 = 2 d_1$, and $d_3 = 3 d_1$.} \label{tab:ex1b}
\end{table}
\begin{table}[htbp]
\centering
{\footnotesize
\begin{tabular}{cclllll|lllll|lllll}
&& \multicolumn{5}{c}{SIS} & \multicolumn{5}{c}{SIRS} & \multicolumn{5}{c}{DC-SIS} \\
\cline{4-6} \cline{9-11} \cline{14-16}
&& \multicolumn{4}{c}{$\mathcal{P}_s$} & $\mathcal{P}_a$ &
\multicolumn{4}{c}{$\mathcal{P}_s$} & $\mathcal{P}_a$ &
\multicolumn{4}{c}{$\mathcal{P}_s$} & $\mathcal{P}_a$ \\
Model & Size & $X_1$ & $X_2$ & $X_{12}$ & $X_{22}$ & All
& $X_1$ & $X_2$ & $X_{12}$ & $X_{22}$ & All
& $X_1$ & $X_2$ & $X_{12}$ & $X_{22}$ & All \\
\hline
 && \multicolumn{15}{c}{Case 3: $p=5000$ and $\sigma_{ij}=0.5^{|i-j|}$}  \\
(1.a)&$d_1$& 1.00& 0.96& 1.00& 1.00& 0.96& 1.00& 0.96& 1.00& 1.00& 0.96& 1.00& 0.96& 1.00& 1.00& 0.96\\
     &$d_2$& 1.00& 0.96& 1.00& 1.00& 0.96& 1.00& 0.96& 1.00& 1.00& 0.96& 1.00& 0.96& 1.00& 1.00& 0.96\\
     &$d_3$& 1.00& 0.97& 1.00& 1.00& 0.96& 1.00& 0.97& 1.00& 1.00& 0.96& 1.00& 0.96& 1.00& 1.00& 0.96\\
(1.b)&$d_1$& 0.09& 0.09& 1.00& 1.00& 0.05& 0.04& 0.03& 1.00& 1.00& 0.01& 1.00& 1.00& 1.00& 1.00& 1.00\\
     &$d_2$& 0.13& 0.13& 1.00& 1.00& 0.06& 0.05& 0.05& 1.00& 1.00& 0.01& 1.00& 1.00& 1.00& 1.00& 1.00\\
     &$d_3$& 0.15& 0.15& 1.00& 1.00& 0.08& 0.07& 0.07& 1.00& 1.00& 0.02& 1.00& 1.00& 1.00& 1.00& 1.00\\
(1.c)&$d_1$& 0.09& 0.10& 0.01& 1.00& 0.00& 0.03& 0.03& 0.99& 1.00& 0.01& 1.00& 1.00& 0.99& 1.00& 0.99\\
     &$d_2$& 0.12& 0.14& 0.03& 1.00& 0.00& 0.06& 0.05& 1.00& 1.00& 0.02& 1.00& 1.00& 1.00& 1.00& 1.00\\
     &$d_3$& 0.14& 0.15& 0.04& 1.00& 0.00& 0.06& 0.07& 1.00& 1.00& 0.02& 1.00& 1.00& 1.00& 1.00& 1.00\\
(1.d)&$d_1$& 0.77& 0.47& 0.52& 0.53& 0.14& 1.00& 0.95& 1.00& 0.04& 0.04& 1.00& 0.93& 1.00& 1.00& 0.92\\
     &$d_2$& 0.80& 0.54& 0.58& 0.57& 0.19& 1.00& 0.95& 1.00& 0.08& 0.07& 1.00& 0.93& 1.00& 1.00& 0.93\\
     &$d_3$& 0.82& 0.57& 0.60& 0.60& 0.21& 1.00& 0.96& 1.00& 0.10& 0.10& 1.00& 0.93& 1.00& 1.00& 0.93\\
 && \multicolumn{15}{c}{Case 4: $p=5000$ and $\sigma_{ij}=0.8^{|i-j|}$}  \\
(1.a)&$d_1$& 0.99& 0.98& 0.98& 1.00& 0.95& 0.99& 0.98& 0.98& 1.00& 0.95& 0.99& 0.98& 1.00& 1.00& 0.97\\
     &$d_2$& 0.99& 0.98& 0.99& 1.00& 0.96& 0.99& 0.99& 0.99& 1.00& 0.96& 0.99& 0.99& 1.00& 1.00& 0.98\\
     &$d_3$& 0.99& 0.99& 0.99& 1.00& 0.97& 0.99& 0.99& 0.99& 1.00& 0.97& 0.99& 0.99& 1.00& 1.00& 0.98\\
(1.b)&$d_1$& 0.62& 0.71& 0.98& 1.00& 0.61& 0.52& 0.63& 0.99& 1.00& 0.51& 1.00& 1.00& 1.00& 1.00& 1.00\\
     &$d_2$& 0.68& 0.76& 0.99& 1.00& 0.67& 0.60& 0.70& 0.99& 1.00& 0.60& 1.00& 1.00& 1.00& 1.00& 1.00\\
     &$d_3$& 0.71& 0.78& 0.99& 1.00& 0.70& 0.62& 0.73& 0.99& 1.00& 0.62& 1.00& 1.00& 1.00& 1.00& 1.00\\
(1.c)&$d_1$& 0.11& 0.11& 0.91& 1.00& 0.06& 0.11& 0.13& 1.00& 1.00& 0.08& 1.00& 1.00& 1.00& 1.00& 1.00\\
     &$d_2$& 0.16& 0.16& 0.93& 1.00& 0.10& 0.20& 0.24& 1.00& 1.00& 0.16& 1.00& 1.00& 1.00& 1.00& 1.00\\
     &$d_3$& 0.20& 0.21& 0.94& 1.00& 0.13& 0.22& 0.30& 1.00& 1.00& 0.19& 1.00& 1.00& 1.00& 1.00& 1.00\\
(1.d)&$d_1$& 0.70& 0.61& 0.50& 0.49& 0.14& 0.98& 0.98& 0.99& 0.18& 0.17& 0.97& 0.96& 1.00& 1.00& 0.94\\
     &$d_2$& 0.73& 0.65& 0.55& 0.53& 0.18& 0.99& 0.98& 0.99& 0.25& 0.24& 0.98& 0.97& 1.00& 1.00& 0.95\\
     &$d_3$& 0.76& 0.67& 0.57& 0.57& 0.21& 0.99& 0.98& 0.99& 0.30& 0.28& 0.98& 0.97& 1.00& 1.00& 0.95
\end{tabular}
}
\caption{The proportions of $\mathcal{P}_s$ and $\mathcal{P}_a$ in our example. This is for the remaining
two cases.
The user-specified model sizes are $d_1 = [p/ 10\log n ]$, $d_2 = 2 d_1$, and $d_3 = 3 d_1$.} \label{tab:ex1b-2}
\end{table}
Tables \ref{tab:ex1a}, \ref{tab:ex1b}, and \ref{tab:ex1b-2} present the simulation results for
$\mathcal{S}, \mathcal{P}_s$, and
$\mathcal{P}_a$.
The performances of the DC-SIS, SIS, and SIRS are quite
similar in model (1.a), indicating that the SIS has a robust performance
if the working linear model does not deviate far from
the underlying true model. The DC-SIS outperforms the SIS and
the SIRS significantly in models (1.b)-(1.d). Both the SIS and
the SIRS have little chance to identify the important predictors
$X_1$ and $X_2$ in models (1.b) and (1.c), and $X_{22}$ in model (1.d).

Comparing Tab.s \ref{tab:ex1a} and \ref{tab:ex1b} with the counterparts in \citeasnoun{Li2012},
one can clearly see that the advantage of using the distance correlation becomes more evident,
observing smaller sample quantiles of $\mathcal{S}$ for DC-SIS, and larger coverage probabilities in $\mathcal{P}_s$
and $\mathcal{P}_a$.

\section{Conclusion}
\label{sec:conclude}

Distance correlation has been found useful in many applications \cite{Kong2012,Li2012}.
A direct implementation of the distance correlation led to an $O(n^2)$ algorithm with sample size $n$.
We propose a fast algorithm.
Its computational complexity is $O(n \log n)$ on average.
Armed with this fast algorithm, we carry out some numerical experiments with sample sizes that have not
been attempted before.
We found that in many cases, the advantage of adopting the distance correlation becomes even more evident.
The proposed fast algorithm certainly makes the distance correlation more applicable in situations
where statistical dependence needs to be evaluated.

\appendix

\makeatletter   %% HAVE TO ADD SOMETHING HERE TO MAKE IT SAY "APPENDIX"
 \renewcommand{\@seccntformat}[1]{APPENDIX~{\csname the#1\endcsname}.\hspace*{1em}}
 \makeatother

\section{Algorithms}

Algorithm \ref{alg:FaDCor} is the algorithm that can compute for
the distance covariance at O($n \log n$).
Algorithm \ref{alg:PartialSum2D} realizes the idea that is described in the proof
of Lemma \ref{lem:gamma}.
Algorithm \ref{alg:sub01} is a subroutine that will be called in Algorithm \ref{alg:PartialSum2D}.

\begin{algorithm}[htbp]\label{alg:FaDCor}
\begin{center}
Algorithm: Fast Computing for Distance Covariance (FaDCor)
\end{center}
{\bf Inputs:} Observations $x_1,\ldots,x_n$, and $y_1,\ldots,y_n$. \\
{\bf Outputs:} The distance covariance that was defined in \eqref{eq:Omega_n}.
\begin{enumerate}
\item Sort $x_1,\ldots,x_n$, and $y_1,\ldots,y_n$.
 Let $I^x$ and $I^y$ denote the order indices; i.e., if for $i, 1\le i \le n$, $I^x(i)=k$,
 then $x_i$ is the $k$th smallest observations among $x_1,\ldots,x_n$. \\
 Similarly if for $i, 1\le i \le n$, $I^y(i)=k$,
 then $y_i$ is the $k$th smallest observations among $y_1,\ldots,y_n$.

\item Let $x_{(1)} < \cdots < x_{(n)}$, and $y_{(1)} < \cdots < y_{(n)}$ denote the order statistics. \\
 Denote the partial sums:
 $$
 s^x(i) =\sum_{j=1}^i x_{(j)}, \quad s^y(i) =\sum_{j=1}^i y_{(j)}, \quad i=1,\ldots, n.
 $$
 They can be computed using the following recursive relation:
 $s^x(1)=x_{(1)}, s^y(1)=y_{(1)}$,
 $$
 s^x(i+1) =s^x(i) + x_{(i+1)}, \quad s^y(i+1) =s^y(i) + y_{(i+1)}, \mbox{ for } i=1,\ldots,n-1.
 $$

\item Compute $\alpha_i^x$, $\alpha_i^y$,
  $\beta_i^x$, and $\beta_i^y$ that are defined in Lemma \ref{lem:AiDot}
  and \ref{lem:BiDot}, using the following formula: for $i=1,\ldots,n$,
  we have
  \begin{eqnarray*}
  \alpha_i^x = I^x(i)-1, && \alpha_i^y = I^y(i)-1, \\
  \beta_i^x = s^x(I^x(i)-1), && \beta_i^y = s^y(I^y(i)-1).
  \end{eqnarray*}

\item Compute $x_\cdot$ and $y_\cdot$ per their definitions in Lemma \ref{lem:AiDot}
  and \ref{lem:BiDot}.

\item Using \eqref{eq:aidot} and \eqref{eq:bidot},
  compute $\sum_{i-1}^n a_{i\cdot}b_{i\cdot}$. \label{alg:1.part2}

\item Using \eqref{eq:adotdot} and \eqref{eq:bdotdot},
  compute $a_{\cdot\cdot}$ and $b_{\cdot\cdot}$. \label{alg:1.part3}

\item Use Algorithm {\em PartialSum2D} to compute for $\gamma_i(\{1\})$, $\gamma_i(\{x_j y_j\})$, $\gamma_i(\{y_j\})$, and $\gamma_i(\{x_j\})$.

\item Using \eqref{eq:sumab} to compute $\sum_{i \neq j} a_{ij}b_{ij}$. \label{alg:1.part1}

\item Finally, apply the results of steps \ref{alg:1.part2}., \ref{alg:1.part3}., and \ref{alg:1.part1}.
  to \eqref{eq:Omega_n}.

\end{enumerate}
\caption{The O($n \log n$) algorithm to compute for the distance covariances.}
\end{algorithm}

\begin{algorithm}[htbp]\label{alg:PartialSum2D}
\begin{center}
Algorithm: Fast Algorithm for a $2$-D Partial Sum Sequence (PartialSum2D)
\end{center}
{\bf Inputs:} Observations $x_1,\ldots,x_n$, $y_1,\ldots,y_n$,
and $c_1,\ldots,c_n$. \\
{\bf Outputs:} Quantity $\gamma_i(\{c_j\}) = \sum_{j: j \neq i} c_j S_{ij}$
that is defined in Lemma \ref{lem:aijbij}.
\begin{enumerate}
\item Compute for the order statistics $x_{(1)} <\cdots <x_{(n)}$ for $x_1,\ldots,x_n$.
  Then rearrange triplets $(x_i,y_i,c_j)$'s such that we have $x_1 <\cdots <x_n$.
  Each triplet $(x_i,y_i,c_j)$ ($1 \le i \le n$) stay unchanged.

\item Let $y_{(1)} <\cdots <y_{(n)}$ denote the order statistics for $y_1,\ldots,y_n$, and
  assume that $I^y(i), i=1,2,\ldots,n$, are the order indices; i.e., if $I^y(i)=k$, then $y_i$ is the
  $k$-th smallest among $y_1,\ldots,y_n$.
  Without loss of generality, we may assume that $y_i = I^y(i)$.

\item Evidently aforementioned function $I^y(i)$ is invertible. Let $(I^y)^{-1}(j)$ denote its inverse.
  Define the partial sum sequence: for $1 \le i \le n$,
  $$
  s^y(i) = \sum_{j=1}^i c_{(I^y)^{-1}(j)}.
  $$
  The following recursive relation enables an O($n$) algorithm to compute for all $s^y(i)$'s,
  $$
  s^y(1) = c_{(I^y)^{-1}(1)},\quad s^y(i+1) = s^y(i) + c_{(I^y)^{-1}(i+1)}, \mbox{ for } i \ge 1.
  $$

\item For $1 \le i \le n$, define
  $$
  s^x(i) = \sum_{j=1}^i c_j.
  $$
  Again the above partial sums can be computed in O($n$) steps.

\item Compute $c_\cdot = \sum_{j=1}^n c_j.$

\item Call Subroutine {\em DyadUpdate} to compute for $\sum_{j: j<i, y_j < y_i} c_j $
  for all $i, 1 \le i \le n$.

\item By \eqref{eq:gamma}, we have that
  $$
  \gamma_i(\{c_j\}) = c_\cdot - c_i -2 s^y(i) -2 s^x(i)
  +4 \sum_{j: j < i, y_j < y_i} c_j.
  $$

\end{enumerate}
\caption{A subroutine that will be needed in the fast algorithm for the distance covariance.
This algorithm realizes the ideas in the proof of Lemma \ref{lem:gamma}.}
\end{algorithm}

\begin{algorithm}[htbp]\label{alg:sub01}
\begin{center}
Subroutine: A Dyadic Updating Scheme (DyadUpdate)
\end{center}
{\bf Inputs:} Sequence $y_1,\ldots,y_n$ and $c_1,\ldots,c_n$, where
$y_1,\ldots,y_n$ is a permutation of $\{1,\ldots,n\}$. \\
{\bf Outputs:} Quantities $\gamma_i := \sum_{j: j < i, y_j < y_i} c_j$, $i=1,2,\ldots,n$.
\begin{enumerate}
\item Recall that we have assumed $n=2^L$.
  If $n$ is not dyadic, we simply choose the smallest $L$ such that $n<2^L$.
  Recall that for $\ell = 0,1,\ldots, L-1$, $k=1,2,\ldots, 2^{L-\ell}$, we define a
  close interval
  $$
  I(\ell, k) := [(k-1)\cdot 2^\ell+1, \ldots, k \cdot 2^\ell].
  $$

\item Assign $s(\ell,k) = 0, \forall \ell, k$, and $\gamma_1 = 0$.

\item For $i=2,\ldots, n$, we do the following.
\begin{enumerate}
\item Fall all $(\ell,k)$'s, such that $y_{i-1} \in I(\ell,k)$. Then for these $(\ell,k)$'s, do update
  $$
  s(\ell,k) \leftarrow s(\ell,k) + c_{i-1}.
  $$
\item Find nonnegative integers $\ell_1 > \cdots > \ell_\tau \ge 0$ such that
  $$
  y_i - 1 = 2^{\ell_1} + \cdots + 2^{\ell_\tau}.
  $$
  Let $k_1 =1$. For $j=2,\ldots,\tau$, compute
  $$
  k_j = (2^{\ell_1} + \cdots + 2^{\ell_{j-1}}) \cdot 2^{-\ell_j} + 1.
  $$

\item Compute $\gamma_i = \sum_{j=1}^\tau s(\ell_j, k_j).$
\end{enumerate}
\end{enumerate}
\caption{A subroutine that will be called in Algorithm \ref{alg:PartialSum2D}.}
\end{algorithm}

\section{Proofs}

\begin{proof}[Proof of Lemma \ref{lem:Omega1}]
One can verify the following equalities:
\begin{eqnarray}
&& \sum_{i\neq j} a_{ij}b_{i\cdot} = \sum_{i=1}^n a_{i\cdot}b_{i\cdot},
\sum_{i\neq j} a_{ij}b_{\cdot j} = \sum_{j=1}^n a_{\cdot j}b_{\cdot j},
\sum_{i\neq j} b_{ij}a_{i\cdot} = \sum_{i=1}^n a_{i\cdot}b_{i\cdot},
\sum_{i\neq j} b_{ij}a_{\cdot j} = \sum_{j=1}^n a_{\cdot j}b_{\cdot j};
\label{eq:check04} \\
&& \sum_{i\neq j} a_{i \cdot} = (n-1)a_{\cdot \cdot}, \quad
\sum_{i\neq j} b_{i \cdot} = (n-1)b_{\cdot \cdot};
\label{eq:check05} \\
&& a_{i \cdot} = a_{\cdot i}, \mbox{ and } b_{i \cdot} = b_{\cdot i}.
\label{eq:check06}
\end{eqnarray}
The following will be used in our simplification too.
We have
\begin{eqnarray}
\sum_{i\neq j} a_{i \cdot} b_{\cdot j} &=&
\sum_{i=1}^n a_{i \cdot} \sum_{j: j\neq i} b_{\cdot j} \nonumber \\
&=& \sum_{i=1}^n a_{i \cdot} (b_{\cdot \cdot} - b_{\cdot i}) \nonumber \\
&=& a_{\cdot \cdot} b_{\cdot \cdot} - \sum_{i=1}^n a_{i \cdot} b_{\cdot i} \nonumber \\
&\stackrel{\mbox{\eqref{eq:check06}}}{=}& a_{\cdot \cdot} b_{\cdot \cdot} - \sum_{i=1}^n a_{i \cdot} b_{i \cdot};
\label{eq:check07}
\end{eqnarray}
Similarly, we have
\begin{equation}\label{eq:check08}
\sum_{i\neq j} b_{i \cdot} a_{\cdot j} =
a_{\cdot \cdot} b_{\cdot \cdot} - \left( \sum_{i=1}^n a_{i \cdot} b_{i \cdot} \right).
\end{equation}

In the following, we simplify the statistic in \eqref{eq:tAdottB}.
We have
\begin{eqnarray*}
\Omega_n &\stackrel{\mbox{\eqref{eq:tAdottB}}}{=}&
\frac{1}{n(n-3)} \sum_{i \neq j}
\widetilde{A}_{i,j} \widetilde{B}_{i,j} \\
&\stackrel{\mbox{\eqref{eq:Aij}}}{=}& \frac{1}{n(n-3)} \sum_{i \neq j}
\left(a_{ij} - \frac{a_{i\cdot}}{n-2} -\frac{a_{\cdot j}}{n-2}
+\frac{a_{\cdot \cdot}}{(n-1)(n-2)} \right)\\
&& \qquad \qquad \qquad
\left(b_{ij} - \frac{b_{i\cdot}}{n-2} -\frac{b_{\cdot j}}{n-2}
+\frac{b_{\cdot \cdot}}{(n-1)(n-2)} \right) \\
&=& \frac{1}{n(n-3)} \sum_{i \neq j} \left[a_{ij}b_{ij}
-\frac{a_{ij}(b_{i\cdot} + b_{\cdot j} )}{n-2}
-\frac{b_{ij}(a_{i\cdot} + a_{\cdot j} )}{n-2}
+\frac{(a_{i\cdot} + a_{\cdot j}) (b_{i\cdot} + b_{\cdot j})}{(n-2)^2}\right. \\
&& \qquad \qquad \qquad \left.
+\frac{a_{ij}b_{\cdot \cdot} + b_{ij}a_{\cdot \cdot}}{(n-1)(n-2)}
-\frac{(a_{i \cdot} + a_{\cdot j})b_{\cdot \cdot} + (b_{i \cdot} + b_{\cdot j})a_{\cdot \cdot} }{(n-1)(n-2)^2}
+ \frac{a_{\cdot \cdot} b_{\cdot \cdot}}{(n-1)^2 (n-2)^2} \right].
\end{eqnarray*}
Furthermore, we have
\begin{eqnarray*}
\Omega_n
&\stackrel{\mbox{\eqref{eq:check05}}}{=}&
\frac{1}{n(n-3)} \sum_{i \neq j} a_{ij}b_{ij} \\
&&-\frac{1}{n(n-2)(n-3)} \sum_{i \neq j} \left[a_{ij}(b_{i\cdot} + b_{\cdot j} )
 + b_{ij}(a_{i\cdot} + a_{\cdot j} ) \right] \\
&& + \frac{1}{n(n-2)^2 (n-3)} \sum_{i \neq j} (a_{i\cdot} + a_{\cdot j})(b_{i\cdot} + b_{\cdot j})\\
&& -\frac{a_{\cdot \cdot} b_{\cdot \cdot}}{(n-1)(n-2)^2(n-3)} \\
&\stackrel{\mbox{\eqref{eq:check04},\eqref{eq:check06}}}{=}&
\frac{1}{n(n-3)} \sum_{i \neq j} a_{ij}b_{ij}
-\frac{4}{n(n-2)(n-3)} \sum_{i=1}^n a_{i\cdot}b_{i\cdot} \\
&& + \frac{1}{n(n-2)^2 (n-3)} \sum_{i \neq j} (a_{i\cdot} + a_{\cdot j})(b_{i\cdot} + b_{\cdot j})
 -\frac{a_{\cdot \cdot} b_{\cdot \cdot}}{(n-1)(n-2)^2(n-3)}.
\end{eqnarray*}
Now bringing in \eqref{eq:check07} and \eqref{eq:check08}, we have
\begin{eqnarray*}
\Omega_n &=&
\frac{1}{n(n-3)} \sum_{i \neq j} a_{ij}b_{ij}
-\frac{4}{n(n-2)(n-3)} \sum_{i=1}^n a_{i\cdot}b_{i\cdot}
-\frac{a_{\cdot \cdot} b_{\cdot \cdot}}{(n-1)(n-2)^2(n-3)} \\
&& + \frac{1}{n(n-2)^2 (n-3)} \left[2(n-1)\sum_{i=1}^n a_{i \cdot}b_{i \cdot}
+2\left(a_{\cdot\cdot} b_{\cdot\cdot} - \sum^n_{i=1} a_{i\cdot}b_{i \cdot} \right) \right] \\
&=& \frac{1}{n(n-3)} \sum_{i \neq j} a_{ij}b_{ij}
-\frac{2}{n(n-2)(n-3)} \sum_{i=1}^n a_{i\cdot}b_{i\cdot}
+ \frac{a_{\cdot \cdot} b_{\cdot \cdot}}{n(n-1)(n-2)(n-3)},
\end{eqnarray*}
which is \eqref{eq:Omega_n}.
\end{proof}

\begin{proof}[Proof of Lemma \ref{lem:check02}]
We use arithmetic induction.
Suppose $n=r+1$, \eqref{eq:check01} becomes
$$
(r+1)U_{r+1,r}(x_1,\ldots,x_{r+1}) = \sum_{i=1}^{r+1} U^{-i}_{r+1,r}(x_1,\ldots,x_{r+1}).
$$
By defining $h(x_1,\ldots,x_{i-1},x_{i+1},\ldots,x_{r+1}) = U^{-i}_{r+1,r}(x_1,\ldots,x_{r+1})$, we can verify that $h(\cdot)$ is a kernel function with $r$ variables.
Consequently, $U_{r+1,r}(x_1,\ldots,x_{r+1})$ can be written as \eqref{eq:U-stat}.

Now suppose for any $n \ge n'$, $U_{nr}(x_1,\ldots,x_n)$ has the form as in \eqref{eq:U-stat}, with the function $h(\cdot)$ that was defined above.
Applying \eqref{eq:check01} with $n=n'+1$, we can show that $U_{n'+1,r}(x_1,\ldots,x_{n'+1})$ still has the form as in \eqref{eq:U-stat}, with the same function $h(\cdot)$ that was defined above.
We omit further details.
\end{proof}

\begin{proof}[Proof of Theorem \ref{th:u-stat}]
It is evident to verify that the followings are true: for $i\neq k$,
\begin{eqnarray*}
a_{i\cdot}^{-k} &=& a_{i\cdot} - a_{ik}, \\
b_{i\cdot}^{-k} &=& b_{i\cdot} - b_{ik}, \\
a_{\cdot \cdot}^{-k} &=& a_{\cdot \cdot} - a_{\cdot k} - a_{k \cdot}
 = a_{\cdot \cdot} -2 a_{\cdot k}, \\
b_{\cdot \cdot}^{-k} &=& b_{\cdot \cdot} -2 b_{\cdot k}.
\end{eqnarray*}
We then have
\begin{eqnarray*}
\Omega_{n-1}^{-k} &=&
\frac{\sum_{i \neq j, i\neq k, j\neq k} a_{ij}b_{ij}}{(n-1)(n-4)}
-\frac{2 \sum_{i=1, i\neq k}^n (a_{i\cdot} - a_{ik}) (b_{i\cdot} - b_{ik})}{(n-1)(n-3)(n-4)}
 \nonumber \\
&& + \frac{(a_{\cdot \cdot} -2 a_{\cdot k}) (b_{\cdot \cdot} -2 b_{\cdot k})}{(n-1)(n-2)(n-3)(n-4)}.
\end{eqnarray*}
For the right hand side of \eqref{eq:check03}, we have the following:
\begin{eqnarray*}
\sum_{k=1}^n \Omega_{n-1}^{-k} &=&
\sum_{k=1}^n
\frac{\sum_{i \neq j, i\neq k, j\neq k} a_{ij}b_{ij}}{(n-1)(n-4)}
-\sum_{k=1}^n
\frac{2 \sum_{i=1, i\neq k}^n (a_{i\cdot} - a_{ik}) (b_{i\cdot} - b_{ik})}{(n-1)(n-3)(n-4)} \\
&& + \sum_{k=1}^n
\frac{(a_{\cdot \cdot} -2 a_{\cdot k}) (b_{\cdot \cdot} -2 b_{\cdot k})}{(n-1)(n-2)(n-3)(n-4)} \\
&=& \frac{(n-2) \sum_{i \neq j} a_{ij}b_{ij}}{(n-1)(n-4)}
-\frac{2 \left[(n-3)\sum^n_{i=1} a_{i\cdot} b_{i\cdot}
+ \sum_{i\neq k} a_{ik}b_{ik} \right]}{(n-1)(n-3)(n-4)} \\
&& + \frac{(n-4)a_{\cdot \cdot} b_{\cdot \cdot} +4\sum^n_{k=1} a_{k\cdot} b_{k \cdot}}{(n-1)(n-2)(n-3)(n-4)} \\
&=& \frac{ \sum_{i \neq j} a_{ij}b_{ij}}{n-3}
- \frac{2}{(n-2)(n-3)} \sum^n_{i=1} a_{i\cdot} b_{i \cdot}
+ \frac{a_{\cdot \cdot} b_{\cdot \cdot}}{(n-1)(n-2)(n-3)}.
\end{eqnarray*}
Compare with \eqref{eq:Omega_n}, we can verify that the above equates to
$n\cdot \Omega_n$, which (per Theorem \ref{th:check01}) indicates that $\Omega_n$ is a U-statistic.
The kernel function of the corresponding U-statistic is the inner product that was defined in
\eqref{eq:tAdottB} with $n=4$.
\end{proof}

\begin{proof}[Proof of Lemma \ref{lem:AiDot}]
We have
\begin{eqnarray*}
a_{i\cdot}  &=& \sum^{n}_{\ell=1} a_{i,\ell} = \sum^{n}_{\ell=1} |x_i - x_\ell| \\
&=& \sum_{x_\ell < x_i} (x_i - x_\ell) + \sum_{x_\ell > x_i} (x_\ell - x_i) \\
&=& x_i \left(\sum_{x_\ell < x_i} 1 - \sum_{x_\ell > x_i} 1 \right) - \sum_{x_\ell < x_i} x_\ell
  + \sum_{x_\ell > x_i} x_\ell.
\end{eqnarray*}
It is easy to verify that
$$
\sum_{x_\ell > x_i} 1 = n-1 -\alpha_i^x,
$$
and
$$
\sum_{x_\ell > x_i} x_\ell = x_\cdot - x_i - \beta_i^x.
$$
Taking into account the above two equations, we have
\begin{eqnarray*}
a_{i\cdot}  &=& (2\alpha_i^x - n +1 )x_i -\beta_i^x + x_\cdot - x_i - \beta_i^x \\
&=& x_\cdot + (2\alpha_i^x - n )x_i -2 \beta_i^x,
\end{eqnarray*}
which is \eqref{eq:aidot}.
\end{proof}

\begin{proof}[Proof of Lemma \ref{lem:gamma}]
Without loss of generality (WLOG), we assume that $x_1 < x_2 < \cdots < x_n$.
We have
\begin{eqnarray*}
\gamma_i(\{c_j\}) &=& \sum_{j: j \neq i} c_j S_{ij} \\
&=& \sum_{j: j > i, y_j > y_i} c_j
 + \sum_{j: j < i, y_j < y_i} c_j
 - \sum_{j: j > i, y_j < y_i} c_j
 - \sum_{j: j < i, y_j > y_i} c_j.
\end{eqnarray*}
Note that we can verify the following equations:
\begin{eqnarray*}
\sum_{j: j < i, y_j < y_i} c_j + \sum_{j: j > i, y_j < y_i} c_j
&=& \sum_{j: y_j < y_i} c_j, \\
\sum_{j: j < i, y_j < y_i} c_j + \sum_{j: j < i, y_j > y_i} c_j
&=& \sum_{j: j < i} c_j, \\
\sum_{j: j > i, y_j > y_i} c_j
 + \sum_{j: j < i, y_j < y_i} c_j
 + \sum_{j: j > i, y_j < y_i} c_j
 + \sum_{j: j < i, y_j > y_i} c_j
&=& \sum_{j: j \neq i} c_j = c_\cdot - c_i,
\end{eqnarray*}
where $c_\cdot = \sum_{j=1}^n c_j$.
We can rewrite $\gamma_i(\{c_j\})$ as follows:
\begin{equation}\label{eq:gamma}
\gamma_i(\{c_j\}) = c_\cdot - c_i -2 \sum_{j: y_j < y_i} c_j -2\sum_{j: j < i} c_j
+4\sum_{j: j < i, y_j < y_i} c_j.
\end{equation}
We will argue that the three summations on the right hand side can be implemented by
O$(n \log n)$ algorithms.
First, term $\sum_{j: j < i} c_j$ is a formula for partial sums.
It is known that an O($n$) algorithm exists, by utilizing the relation:
$$
\sum_{j: j < i+1} c_j = c_i + \sum_{j: j < i} c_j .
$$
Second, after sorting $y_j$'s at an increasing order, sums $\sum_{j: y_j < y_i} c_j$ is
transferred into a partial sums sequence.
Hence it can be implemented via an O($n$) algorithm.
If QuickSort \cite{Hoare1961} \cite[Section 5.2.2: Sorting by Exchanging (pages 113-122)]{QuickSort}
is adopted, the sorting of $y_j$'s can be done via an O($n \log n$) algorithm.

We will argue that sums $\sum_{j: j < i, y_j < y_i} c_j, i=1,\ldots,n,$
can be computed in an O($n \log n$) algorithm.
WLOG, we assume that $y_i, i=1,2,\ldots, n$, is a permutation of the set $\{1,2,\ldots, n\}$.
WLOG, we assume that $n$ is dyadic; i.e., $n=2^L$, where $L \in \mathbb{N}$ or $L$ is a nonnegative integer.
For $\ell = 0,1,\ldots, L-1$, $k=1,2,\ldots, 2^{L-\ell}$, we define an close interval
$$
I(\ell, k) := [(k-1)\cdot 2^\ell+1, \ldots, k \cdot 2^\ell].
$$
We then define the following function
$$
s(i,\ell,k) := \sum_{j: j<i, y_j \in I(\ell, k)} c_j,
$$
where $i=1,\ldots,n$, $\ell = 0,1,\ldots, L-1$, and $k=1,2,\ldots, 2^{L-\ell}$.

We argue that computing the values of $s(i,\ell,k)$ for all $i,\ell,k$, can be done in O($n \log n$).
First of all, it is evident that for all $\ell,k$, we have
$$
s(1,\ell,k) \equiv 0.
$$
Suppose for all $i' \le i$, $s(i',\ell,k)$'s have been computed for all $\ell$ and $k$.
For each $0 \le \ell \le L-1 < \log_2 n$, there is only one $k^\ast$, such that $y_i \in I[\ell, k^\ast]$.
By the definition of $s(\cdot, \cdot, \cdot)$, we have
\begin{eqnarray*}
s(i+1,\ell,k) = \left\{\begin{array}{ll}
s(i,\ell,k) + c_i, & \mbox{ if } k = k^\ast, \\
s(i,\ell,k), & \mbox{ otherwise. }
\end{array}
 \right.
\end{eqnarray*}
The above dynamic programming style updating scheme needs to be run for $n$ times (i.e., for all $1 \le i
\le n$), however each stage requires no more than $\log_2 n$ updates.
Overall, the computing for all $s(i,\ell,k)$ takes no more than O($n \log n$).

For a fixed $i$, $1 \le i \le n$, we now consider how to compute for $\sum_{j: j < i, y_j < y_i} c_j$.
If $y_i=1$, obviously we have $\sum_{j: j < i, y_j < y_i} c_j = 0$.
For $y_i > 1$, there must be a unique sequence of positive integers $\ell_1 > \ell_2 > \cdots
\> \ell_\tau > 0$, such that
$$
y_i -1 = 2^{\ell_1} + 2^{\ell_2} + \cdots + 2^{\ell_\tau}.
$$
Since $y_i \le n$, we must have $\tau \le \log_2 n$.
We then define $k_\alpha, \alpha=1,\ldots, \tau$ as follows
\begin{eqnarray*}
k_1 &=& 1, \\
k_2 &=& 2^{\ell_1 - \ell_2} + 1, \\
&\vdots& \\
k_\alpha &=& (2^{\ell_1} + \cdots + 2^{\ell_{\alpha-1}})/2^{\ell_\alpha} + 1, \\
&\vdots& \\
k_\tau &=& (2^{\ell_1} + \cdots + 2^{\ell_{\tau-1}})/2^{\ell_\tau} + 1.
\end{eqnarray*}
One can then verify the following: for $2 \le i \le n$,
$$
\sum_{j: j < i, y_j < y_i} c_j = \sum_{\alpha=1}^\tau s(i, \ell_\alpha, k_\alpha).
$$
Since $\tau \le \log_2 n$, the above takes no more than O($\log n$) numerical operations.
Consequently, computing $\sum_{j: j < i, y_j < y_i} c_j$ for all $i, 1 \le i \le n$, can
be done in O($n \log n$).
(We realized that the above approach utilized the AVL tree structure \cite{AVL1962}.)
From all the above, we established the result.
\end{proof}

\begin{proof}[Proof of Theorem \ref{th:sirs}]
We have the following sequence of equations:
\begin{eqnarray*}
\sum_{j=1}^n \left[\sum_{i=1}^n x_i \mathbf{1}(y_i < y_j) \right]^2
&=& \sum_{j=1}^n \left[\sum_{i=1}^n x_i \mathbf{1}(y_i < y_j) \right] \cdot
 \left[\sum_{k=1}^n x_k \mathbf{1}(y_k < y_j) \right] \\
&=& \sum_{j=1}^n \sum_{i=1}^n \sum_{k=1}^n x_i \cdot x_k \cdot \mathbf{1}(y_i < y_j \mbox{ and } y_k < y_j)\\
&=& \sum_{i=1}^n x_i \left[ \sum_{k: y_i \le y_k} x_k \sum_{j=1}^n \mathbf{1}(y_k < y_j)
 + \sum_{k: y_i > y_k} x_k \sum_{j=1}^n \mathbf{1}(y_i < y_j) \right] .
\end{eqnarray*}
The last expression implies the following steps to compute for SIRS$(X,Y)$.
\begin{enumerate}
\item For $k=1,\ldots,n$, compute $\alpha_k := \sum_{j=1}^n \mathbf{1}(y_k < y_j)$;

\item For $i=1,\ldots,n$, compute $\beta_i := \sum_{k: y_k \ge y_i} x_k \alpha_k$;

\item For $i=1,\ldots,n$, compute $\gamma_i := \sum_{k: y_k < y_i} x_k $;

\item Compute
$$
\mbox{SIRS}(X,Y) = \frac{\sum_{i=1}^n x_i (\beta_i + \gamma_i \alpha_i)}{n(n-1)(n-2)}.
$$
\end{enumerate}
Since $\alpha_i$'s, $\beta_i$'s, and $\gamma_i$'s are partial sums,
it is easy to verify that each of the above steps can be done within O$(n \log n)$ operations on average, hence the entire algorithm takes O$(n \log n)$ operations on average.
\end{proof}

%\bibliographystyle{amsplain}%
%\bibliographystyle{plain}  %
%\bibliographystyle{unsrt}    %
%\bibliographystyle{alpha}    %
%\bibliographystyle{abbrv}    %

%\bibliography{quasiMeans}
%% HERE WE DECLARE THE BIBLIOGRAPHYSTYLE TO USE AND THE BIBLIOGRAPHY DATABASE
%\bibliographystyle{ECA_jasa}
%\bibliography{fastDistCorr}

\end{document}